\documentclass[a4paper]{article}

\usepackage{hyperref}
\usepackage{algorithmicx}
\usepackage[utf8]{inputenc}
\usepackage{algorithm}
\usepackage{algpseudocode}
\usepackage[margin=.85in]{geometry}
\usepackage{graphicx}
\usepackage{float}
\usepackage{subfig}
\usepackage{amsthm,color}
\usepackage{amsmath,amssymb}
\usepackage[dvipsnames]{xcolor}
\usepackage{paralist}
\usepackage[title]{appendix}
\usepackage{systeme}
\usepackage{soul}
\usepackage{multirow}
\usepackage{multicol}

\newcounter{casectr}
\newenvironment{case}{
	\begin{list}{{\sc Case \arabic{casectr}.}\enspace}
		{\usecounter{casectr}
			\setlength{\leftmargin}{0in}
			\addtolength{\itemsep}{\parsep}
			\setlength{\parsep}{0in}
			\setlength{\labelwidth}{\leftmargin}
			\addtolength{\labelwidth}{-\labelsep}
			\setlength{\rightmargin}{0in}
			\setlength{\listparindent}{\parindent}
		}
	}{\end{list}}

\newcommand{\rmv}[1]{}

\newcommand{\DNB}[1]{{#1}\text{-DNB}}
\newcommand{\NB}[1]{{#1}\text{-NB}}

\definecolor{purple}{rgb}{0.6,0,0.6}

\definecolor{darkgreen}{rgb}{0,0.4,0}

\newcommand{\bemph}[1]{{\bfseries\itshape#1}}

\newcommand{\Z}{\mathbb{Z}}

\newcommand{\val}{\textit{value}}
\newcommand{\applyT}{\textit{apply}_T}
\newcommand{\ptr}{\textit{ptr}}

\newcommand{\NQ}{\textsc{enqueue}}
\newcommand{\DQ}{\textsc{dequeue}}
\newcommand{\CAS}{\textsc{CAS}}
\newcommand{\afterTail}{\textit{l\_next}}
\newcommand{\emptyq}{\bot}
\newcommand{\gINode}{\textit{g\_Init\_Node}}

\newcommand{\gIDQ}{\textit{g\_Init\_DQ}}
\newcommand{\gHead}{\textit{g\_Head}}
\newcommand{\gTail}{\textit{g\_Tail}}
\newcommand{\gAnnE}{\textit{g\_Ann\_E}}
\newcommand{\gAnnD}{\textit{g\_Ann\_D}} 

\newcommand{\trynq}{\textsc{trytoenqueue}}
\newcommand{\trydq}{\textsc{trytodequeue}}
\newcommand{\add}{\textit{addr}}
\newcommand{\fNext}{\textit{next}}
\newcommand{\fFlag}{\textit{flag}}
\newcommand{\lHead}{\textit{l\_Head}}
\newcommand{\lTail}{\textit{l\_Tail}}
\newcommand{\lAnnE}{\textit{l\_Ann\_E}}
\newcommand{\lAnnD}{\textit{l\_Ann\_D}}
\newcommand{\lR}{\textit{l\_R}}
\newcommand{\gR}{\textit{g\_R}}

\newcommand{\lNode}{\textit{l\_Node}}

\newcommand{\lnext}{\textit{l\_next}}

\newcommand{\ladd}{\textit{l\_addr}}

\newcommand{\mathnull}{\textit{NULL}}

\newcommand{\vDone}{\textbf{done}}
\newcommand{\vFail}{\textbf{failed}}
\newcommand{\TRUE}{\textbf{done}}
\newcommand{\FALSE}{\textbf{failed}}

\newcommand{\res}{\textit{response}}
\newcommand{\gAnn}{\textit{g\_Ann}}

\newcommand{\tnb}{\textsc{2NB-U}}
\newcommand{\stt}{\textit{state}}
\newcommand{\op}{\textit{op}}
\newcommand{\try}{\textsc{TryToDo}}
\newcommand{\ex}{\textit{Exp}}

\newcommand{\rememberlines}{\xdef\rememberedlines{\number\value{ALG@line}}}
\newcommand{\resumenumbering}{\setcounter{ALG@line}{\rememberedlines}}

\algrenewcommand\algorithmicindent{3.7mm}

\newcounter{theorem}
\newcounter{NewCounter}
\newcounter{claimcount}[NewCounter]
\makeatletter
\renewcommand{\p@claimcount}{\theNewCounter.}
\makeatother

\newcounter{cclaimcount}[claimcount]
\makeatletter
\renewcommand{\p@cclaimcount}{\theNewCounter.\theclaimcount.}
\makeatother

\newenvironment{claim}{
	\par\refstepcounter{claimcount}\vspace{5pt}\noindent\textbf{Claim \arabic{NewCounter}.\arabic{claimcount}}\begin{itshape}
	}{\end{itshape}}

\newtheorem{theorem}[NewCounter]{Theorem}

\newtheorem{corollary}[NewCounter]{Corollary}
\newtheorem{lemma}[NewCounter]{Lemma}
\newtheorem{observation}[NewCounter]{Observation}

\newtheorem{definition}[NewCounter]{Definition}
\title{Differentiated nonblocking: a new progress condition\\and a matching queue algorithm}

\author{David Y. C. Chan\footnote{Department of Computer Science, University of Calgary.} \and Shucheng Chi\footnote{Institute for Interdisciplinary Information Sciences, Tsinghua University.}\\ \and Vassos Hadzilacos\footnote{Department of Computer Science, University of Toronto.} \and Sam Toueg\footnote{Department of Computer Science, University of Toronto.}} 

\begin{document}
	\maketitle
	
	\begin{abstract}
		In this paper, we first propose a new liveness requirement for shared objects and data structures,
		%	that we call \emph{differentiated $k$-nonblocking}, or $\DNB{k}$, for short
		we then give a shared queue algorithm that satisfies this requirement and we prove its correctness.
		We also implement this algorithm and compare it to a well-known
		shared queue algorithm that is used in practice~\cite{MichaelScott96}.
		In addition to having a stronger worst-case progress guarantee,
		our experimental results suggest that, at the cost of a marginal decrease in throughput,
		our algorithm is significantly \emph{fairer},
		by a natural definition of fairness that we introduce here.
	\end{abstract}

	\section{Introduction}
	
	In this paper, we first propose a new liveness requirement for shared objects and data structures;
	%	that we call \emph{differentiated $k$-nonblocking}, or $\DNB{k}$, for short
	we then give a shared queue algorithm that satisfies this requirement and we prove its correctness.
	We also implement this algorithm and compare it to a well-known
	shared queue algorithm that is used in practice~\cite{MichaelScott96}.
	In addition to having a stronger worst-case progress guarantee,
	our experimental results suggest that, at the cost of a marginal decrease in throughput,
	our algorithm is much more \emph{fair},
	by a natural definition of fairness that we introduce here. %\RMP{Make sure no one defined it before. Refer to https://dl.acm.org/doi/abs/10.1145/2767386.2767430 and to https://arxiv.org/pdf/1311.3200.pdf}
	%Roughly speaking, a shared object implementation is fair
	%	if each process gets a share of the object
	%	that is proportional to the process relative speed:
	%	e.g., a process $p$ that is twice as fast as a process $q$ gets about twice as many successful operations as $q$,
	%        and if $p'$ that is slower than $q'$ by a factor of say 2/3, then it should get about 2/3 as many successful operations as $q$. 
	We now explain the paper's motivation and contributions in more detail.
	
	\subsection{Background}

	%\textbf{Non-blocking versus wait-free implementations.}
	\emph{Wait-freedom}, a well-known liveness requirement for shared objects~\cite{Herlihy91}, 
	guarantees that \emph{all} the (non-faulty) processes
	make progress, i.e., every process completes every operation that it applies on the object.
	This is a very strong progress property but, unfortunately, it is also typically expensive to achieve:
	wait-free implementations
	often use an expensive ``help'' mechanism whereby
	processes that are fast must perform the operations of slower processes
	to prevent them from starving.
	
	This is why many data structures and shared object implementations in practice are only \emph{non-blocking}
	(also known as \emph{lock-free}). %~\cite{NB}
	Such implementations are often more efficient than wait-free ones, but they
	guarantee only that \emph{one} process makes progress:
	if several processes apply operations to a nonblocking object, the object may reply
	only to the operations of a single process,
	while all the other processes can be stuck waiting for replies to their operations forever.
	%Such extreme executions can be rare in practice,
	%	especially if processes have roughly the same speed~\cite{https://arxiv.org/pdf/1311.3200.pdf}.
	%When processes have different speeds, however,
	%	\emph{unfair} executions (whereby a process gets
	%	a lot less than its fair share of the object relative to its speed) can be common:
	%	this was shown in~\cite{https://dl.acm.org/doi/abs/10.1145/2767386.2767430}
	%	for the \emph{Single CAS Universal (SCU)} class of non-blocking algorithms~\cite{https://arxiv.org/pdf/1311.3200.pdf},
	%%	which is widely used\RMP{this is what that paper claims...}
	%%	to design non-blocking data structures~\cite{https://arxiv.org/pdf/1311.3200.pdf},
	%	and was also seen in our experiments.\footnote{Under a high workload,
	%	we observed that the commonly-used non-blocking queue algorithm of~\cite{}
	%	can almost completely starve a process if it is significantly slower than the other ones.}
	
	More recently, some papers~\cite{BushkovGuerraoui15,Ben-DavidCHT16} considered a parametrized progress requirement,
	called \emph{$k$-nonblocking}, that ranges from non-blocking (when $k=1$) to wait-freedom (when $k=n$, the number of processes in the system):
	intuitively, when $n$ processes apply operations,
	a $k$-nonblocking object is guaranteed to reply to the operations of at least $k$ processes.
	A universal construction given in~\cite{Ben-DavidCHT16} shows that $k$-nonblocking can be achieved with a help mechanism
	whose worst-case cost is proportional to $k$ rather than $n$.
	This suggests that for small $k>1$, $k$-nonblocking implementations could be useful
	in practical settings
	as a compromise between the efficiency of non-blocking and the stronger progress guarantee of wait-freedom.
	
	\subsection{A new progress requirement}
	In this paper, we first identify a weakness in the non-blocking and $k$-nonblocking progress requirements,
	and propose a stronger versions of these two.
	To illustrate this weakness, consider a \emph{non-blocking queue} that is used by two processes: a \emph{producer}
	that repeatedly enqueues items (e.g., jobs to be executed)
	and a \emph{consumer} that repeatedly dequeues items for processing.
	Here, the non-blocking progress guarantee of the queue could be useless:
	by definition, a non-blocking queue allows runs in which % the producer is the only process to ever succeed as follows:
	the producer succeeds in enqueueing every item (of an infinite sequence of items)
	while the consumer is not able to dequeue even one.\footnote{Intuitively, in a non-blocking \emph{implementation},
		this may occur because the producer is faster than the consumer
		and its repeated enqueues interfere with every dequeuing attempt by the slower consumer.}
	So a non-blocking queue may prevent the producer from ever communicating any item to the consumer.
	Non-blocking also allows the symmetric case: the \emph{only}
	process to ever succeed could  be the (fast) consumer that keeps dequeuing $\bot$ (``the queue is empty''),
	while the (slow) producer is not able to enqueue even one item.
	
	It is easy to see that,
	for $k<n$, the $k$-nonblocking liveness requirement,
	which guarantees that up to $k$ processes make progress,
	suffers from the same problem.
	For example, if $k$ producers enqueue items into a $k$-nonblocking queue,
	and one consumer dequeues items from that queue,
	it is possible that only the $k$ enqueuers make progress while
	the consumer may fail to dequeue any item.
	
	The above weakness of the non-blocking and $k$-nonblocking progress requirements
	extends to other shared objects.
	For example, consider a non-blocking or $k$-nonblocking \emph{dictionary},
	with the standard operations of \emph{insert}, \emph{delete} and \emph{search}:
	the liveness requirement does not prevent
	executions where the only operations to complete are, say, \emph{insert}s
	(which intuitively ``write'' the state of the dictionary),
	while every process that attempts to execute a \emph{search} (which ``reads'' the state of the dictionary)
	may be prevented from completing its operation;
	this effectively renders the shared dictionary useless.
	
	It is important to note that this weakness is not
	just permitted by the definition of non-blocking,
	but it is actually exhibited in some existing
	implementations of non-blocking objects.
	For example, consider the simple non-blocking version of the \emph{atomic snapshot} algorithm
	described by Afek et al. in~\cite{AADGMS_93}.
	It is easy to see that this non-blocking algorithm allows
	the writers of the atomic snapshot to be the only ones to ever succeed,
	while all the readers (i.e., scanners) are stuck forever trying to read.
	%This non-blocking atomic snapshot is effectively useless in such worst-case executions.
	
	In general, the utility of an object is given by the set of operation types that it supports.
	By definition, the non-blocking or $k$-nonblocking progress requirements (for $k<n$) allow
	\emph{all but one} of the operation types to be effectively suppressed;
	this may significantly degrades the object's functionality.
	
	In view of the above, we propose a natural strengthening of the $k$-nonblocking progress requirement
	(for $k =1$ this is a strengthening of non-blocking) that 
	ensures that no operation type is ever suppressed;
	we call it \emph{differentiated $k$-nonblocking}, or $\DNB{k}$, for short.
	Intuitively, $\DNB{k}$ guarantees that, for every operation type $T$,
	if a process is stuck while performing an operation of type $T$
	(i.e., it takes an infinite
	number of steps without completing this operation)
	then at least $k$ processes complete infinitely many operations \emph{of type $T$}.
	Note that $\DNB{k}$ guarantees that an object remains fully functional, in the sense that no operation type
	can be suppressed;
	as with $k$-nonblocking, $\DNB{k}$ also guarantees that at least $k$ processes make progress;
	in other words, $\DNB{k}$ implies $k$-nonblocking. \footnote{Note that
		a $\DNB{k}$ object with $d$ types of operations is \emph{incomparable} to a ${kd}$-nonblocking object.}
	
	\subsection{A corresponding algorithm and its proof}
	To demonstrate the feasibility of this new progress property,
	we derive and prove the correctness of a $\DNB{2}$ algorithm for a shared queue: 
	this algorithm guarantees that at least
	two enqueuers \emph{and} at least two dequeuers, all operating concurrently, make progress
	(we sometimes refer to this as the ``2+2 non-blocking'' property).
	We chose to design a queue because of its wide use in distributed computing,
	and we chose $k=2$ because we wanted an algorithm that guarantees progress by more than one process
	while remaining efficient.
	Intuitively, our $\DNB{2}$ queue algorithm is efficient because its
	help mechanism is light-weight:
	(1)~each process makes at most a \emph{single} attempt to help \emph{one}
	other process to do its operation
	before doing its own operation,
	and
	(2)~dequeuers help only dequeuers,
	and enqueuers help only enqueuers.
	In fact, in our  algorithm, enqueuers do not interfere (by helping or competing with) dequeuers, and vice versa.
	%The worst-case number of CAS operations
	%	done for each completed enqueue operation and each completed enqueue operation
	%	is the same as in the well-known
	%	nonblocking queue algorithm by MS~\cite{} (which is used in practice).\RMP{Check this claim DAVID is ``pretty'' sure :-), and David thinks that the enqueue lgo can be improved by halving the number of CAS that it does!}

	In an experimental evaluation of our $\DNB{2}$ queue algorithm, we identified another
	desirable feature of this algorithm, which we call \emph{fairness} and describe below.
	
	\subsection{Fairness}
	%	
	%
	%
	%while the throughput (successful operations per unit of time)
	%	of the $\DNB{2}$ algorithm is 
	%	marginally lower than\MP{??}
	%	the throughput of the MS algorithm,
	%	the ``fairness'' of the $\DNB{2}$ is much higher as we now explain.
	
	It has been observed that non-blocking objects often behave as if they were wait-free, and this is
	because the pathological scenarios that
	lead to blocking are rarely,
	if ever, encountered in practice~\cite{HerlihyS11}.
	Wait-freedom however, does not imply fairness -- a property that may be desirable.
	To illustrate what we mean by fairness,
	suppose that two processes $p$ and $q$ repeatedly apply the same type of operation on a non-blocking object.
	If both processes have similar \emph{speed} (i.e., steps per unit of time) then,
	by symmetry, we would expect both $p$ and $q$ to complete roughly the same number of operations.
	If, however, $p$ and $q$ has significantly different speeds, one of the processes,
	typically the faster one, may be able to complete many more operations (compared to the slower process)
	than its greater speed actually justifies.
	%If, however, $p$ and $q$ has significantly different speeds, the non-blocking object may privilege one of the processes,
	%	typically the faster one, and allows it to complete many more operations (compared to the slower process)
	%	than its greater speed justifies.
	For example, suppose $p$ is twice as fast as $q$. %(i.e., $p$ takes twice as many steps as $q$ in a unit of time).
	In this case the object is fair if it allows $p$ to complete about twice as many~operations~as~$q$;
	but it is unfair if it allows $p$ to do, say, four times as many operations as $q$.
	%So out of, say, 900 completed operations, about $600$ should be of $p$ and about $300$ should be of $q$.
	
	In general, we say that an object (implementation) is \emph{fair}
	if the following holds:
	when $m$ processes
	with speeds $s_1 , s_2, \ldots s_m$,
	repeatedly apply operations of the same type,
	then, for each pair of processes $i$ and $j$,
	the ratio of the number of operations completed by each of the two processes is approximately the ratio
	of their speeds $s_i / s_j$.
	Equivalently, each process $i$
	completes about ${s_i}/({s_1 + s_2 + \ldots + s_m})$ of the total number of operations
	that all the $m$ processes complete; this fraction is what we call the \emph{fair share} of process $i$.
	%In particular, if $p$ is $x$ times as fast as $q$,
	%	the object is fair if it allows $p$ to complete about $x$ times as many operations as $q$.
	%The Single CAS Universal (SCU) family of non-blocking algorithms
	
	As an example of the difference between achieving wait-freedom and achieving fairness,
	consider the \emph{Single CAS Universal (SCU)} class of non-blocking algorithms~\cite{alistarh2016lock}.
	In systems with a stochastic scheduler,
	these algorithms were shown to be wait-free with probability 1.
	This was proven for systems where
	processes have
	the same probability of taking the next step (intuitively, they have the same speed)~\cite{alistarh2016lock},
	and also for systems where
	processes have different probabilities of taking the next step
	(intuitively, they have different speeds)~\cite{alistarh2015lock}.
	But these algorithms
	are
	\emph{not} fair:
	even in systems with only two processes,
	they allow
	a process $p$ that is $x$ times as fast as a process $q$ to complete about
	$x^2$ times as many operations as $q$~\cite{alistarh2015lock},\footnote{With SCU algorithms 
		when there are $m$ processes
		with speeds $s_1 , s_2, \ldots s_m$,
		the expected ratio of the number of operations completed any two processes $i$ and $j$ is
		proportional to the \emph{square} of the ratio
		of their speeds $s_i / s_j$~\cite{alistarh2015lock}.}
	whereas in a fair implementation $p$ would complete about $x$ times as many operations as $q$.
	%	So, the slower process $q$ only gets $(x+1)/(x^2+1)$ of its fair share of operations.
	
	%So if $p$ is $x$ times as fast as $q$,
	%	the object is fair if it allows $p$ to complete about $x$ times as many operations as $q$
	%%The Single CAS Universal (SCU) family of non-blocking algorithms
	%(in this sense, the SCU non-blocking algorithms
	%	defined in~\cite{https://arxiv.org/pdf/1311.3200.pdf}
	%	are not fair because, in expectation, they allow $p$ to complete about
	%	$x^2$ times as many operations as $q$~\cite{https://dl.acm.org/doi/abs/10.1145/2767386.2767430}).
	%More generally, we define an object (implementation) to be fair
	%	if, when $m$ processes
	%	with speeds $s_1 , s_2, \ldots s_m$,
	%	repeatedly apply operations of the same type
	%	then, for each pair of processes $i$ and $j$,
	%	the ratio of the number of operations completed by the two processes is approximately the ratio
	%	of their speeds $s_i / s_j$.\footnote{Equivalently, each process $i$
	%	completes about ${s_i}/({s_1 + s_2 + \ldots + s_m})$ of the total number of operations
	%	that the $m$ processes complete.}

	\subsection{Some experimental results}
	To evaluate the practicality of the $\DNB{2}$ queue algorithm,
	we implemented it and 
	the well-known non-blocking queue algorithm by Michael and Scott (MS)~\cite{MichaelScott96},
	and compared their performance under various metrics.
	The two algorithms differ in their liveness properties:
	the MS algorithm is non-blocking\footnote{It turns out that it is actually a $\DNB{1}$ algorithm,
		although this property had not been defined at that time.}
	and thus does not incur the cost of a help mechanism;
	while our algorithm satisfies the stronger $\DNB{2}$ and employs a light-weight help mechanism.
	
	Our experimental results are only preliminary, and they are only meant to give a rough idea
	on the potential performance of the $\DNB{2}$ algorithm compared to a commonly used algorithm.
	In particular, we did not try to optimize our implementation of the $\DNB{2}$ algorithm,
	%	and we did not deal with memory reclamation;
	and, similarly, we implemented the MS algorithm as written in the original paper %without memory reclamation and
	without optimizations found in existing implementations.
	
	These experimental results strongly suggest that,
	besides the stronger progress guarantee of the $\DNB{2}$ algorithm,
	the $\DNB{2}$ algorithm
	drastically increases fairness across a wide range of process speeds,
	at the cost of a marginal reduction in throughput.
	As expected, both the MS algorithm and our $\DNB{2}$ algorithm are
	fair when all processes have the same speed, but we observed that:
	(1) the fairness of the MS algorithm breaks down
	\emph{when processes have different speeds} and
	it deteriorates rapidly as the differences in speeds increase,
	whereas
	(2) the $\DNB{2}$ algorithm maintains a good level of fairness 
	throughout a wide range of speed differences.
	
	We first considered a system
	with only two enqueuers and two dequeuers %that repeatedly apply their operations,
	where we slowed one of the two enqueuers and one of the two dequeuers  by a factor of $k$,
	for each $k$ in the range $2 \le k < 20$.
	With the MS algorithm, we observed that the slower enqueuer and dequeuer completed much less than their fair share of operations:
	for $k =2$, they got about 45\% and 26\% of their fair share, respectively;
	for $k =3$, they got about 23\% and 8\% of their fair share;
	for $k =5$, they got about 8\% and 1\% of their fair share;  and
	for $k =8$, they got only about 2\% and 0.22\% of their fair share.
	When we reached $k=11$, the slow dequeuer was prevented from completing \emph{any} dequeue operation,
	while the other dequeuer completed about 49,000 operations; and the slow enqueuer
	managed to complete only 43 operations,
	while the other enqueuer completed almost 59,000 operations.
	In sharp contrast, with the $\DNB{2}$ algorithm, even when the slow enqueuer and dequeuer
	were slowed down by a factor of $k=8$,
	they still completed at least 60\% of their fair share of operations.
	%	(with $k=20$ and $k=256$, they completed at least 55\% XXX of their their fair share).\RMP{need experiment with k=20
	%	 and an even more extreme case: we had one where k=256!}
	
	While the fairness of the MS algorithm can be very poor, however,
	its throughput (i.e., the total number of completed operations)
	remained consistently higher than our $\DNB{2}$ algorithm:
	for $k$ ranging from 1 (where all four processes have the same speed) to 19 (where slow processes are 19 times slower than fast ones), the ``enqueue" throughput of the $\DNB{2}$ algorithm was about 67\% to 62\%
	of the throughput of the MS algorithm, while its ``dequeue" throughput was about  88\% to 76\% of the MS algorithm.%\SMP{Shucheng: I've checked the numbers and correct them.}

	Since a $\DNB{2}$ queue is, by definition, wait-free for two enqueuers and two dequeuers,
	one may wonder whether our algorithm's good fairness behaviour
	holds up in systems with more processes.
	To check this, we also run experiments with a system with 8 enqueuers and 8 dequeuers.
	These experiments confirmed that the algorithm remains fair throughout a wide range of process speed differences.
	Moreover, its throughput is closer to the throughput of the MS algorithm in this larger system.

	It is also worth noting that our experiments confirmed that the $\DNB{2}$ algorithm
	has the following desirable feature: the throughput of the enqueuers depends only on their speed,
	it does not depend on the speed or even the presence of dequeuers;
	symmetrically the throughput of dequeuers is independent of the speed or presence of enqueuers.
	For example, in one experiment, we considered a system with two groups, 8 enqueuers and 8 dequeuers,
	where all processes have the same speed. %\SMP{Figure out slowdown anomaly!}
	We first experimented with a run where every process in \emph{both} groups applied operations,
	and we saw that each enqueuer completed about 12000 operations and each dequeuer
	completed about 14500 operations.
	We then experimented with two other runs:
	one with \emph{only} the 8 enqueuers present, and one with \emph{only} the 8 dequeuers present.
	In both cases, each participating process completed about the same number of operations as in the first run:
	namely, 12000 for the enqueuers and 14500 for the dequeuers.
	So the performance of the two groups are indeed independent of each other.
	As we noted earlier, this is because in our algorithm enqueuers do not interfere with dequeuers, and vice-versa.

	The dramatic improvement of fairness exhibited by the $\DNB{2}$ algorithm
	(over the MS algorithm) is not accidental: it is due to the way we
	designed it to guarantee progress by at least two enqueuers and at least two dequeuers, as we now explain.

	%
	%We run some experiments to compare the fairness of the MS algorithm and our $\DNB{2}$ algorithm
	%	under different scenarios.
	%These experiments showed that
	%	while both the MS algorithm and our $\DNB{2}$ algorithm are
	%	indeed fair when all processes have the same speed,
	%	\emph{when processes have different speeds}
	%	the fairness of the MS algorithm deteriorates
	%	rapidly\MP{significantly?} as the differences in speeds increase:
	%	for example, when the speed difference is a factor of 10,
	%	the MS algorithm effectively is no longer ``wait-free in practice'': slow processes
	%	are suppressed to a point where they only do z operations every...
	%In contrast, we observed that even in the face of huge process speed differentials,
	%	 the $\DNB{2}$ algorithm remains fair (among the group of enqueuers, and also the group of dequeuers).
	%
	
	\subsection{The help mechanism and fairness}
	At a high level, the help mechanism for enqueue operations works as follows
	(the help mechanism for dequeue operations is symmetric).
	There is a single register $R_E$ that is shared by all enqueuers; $R_E$ contains at most one call for help by some enqueuer.
	To execute an enqueue operation, a process $p$ first tries \emph{once} to help the process%\MP{No: it sees an anonymous call for help}
	that it sees in $R_E$ (if any).
	Whether $p$ succeeds in this single attempt or not, 
	it then tries to perform its \emph{own} operation.
	If this attempt fails,
	$p$ overwrites any existing call for help in $R_E$ with its own call for help;
	then it tries to perform its own operation again.
	This attempt to perform its own operation, followed by overwriting $R_E$ with its call for help if the attempt fails,
	continues until $p$ succeeds in completing its operation, whether on its own or with the help of another process
	which saw $p$'s call for help in $R_E$.
	Note that with this help mechanism,
	all the enqueuers that need help effectively compete with each other by (over-)writing their call for help in the \emph{same} register $R_E$. %, and overwriting each other.
	We now explain how this helps achieve fairness.
	
	Suppose a process $i$ is twice as fast as a process $j$,
	and that both $i$ and $j$
	repeatedly fail to perform their enqueue operations, and so they repeatedly call for help
	via the shared register $R_E$.
	Then the periods when $R_E$ contains the calls for help by $i$ are approximately twice as long as the periods
	when $R_E$ contains the calls for help by $j$.
	So whenever a process reads $R_E$ to see which process to help, it is twice as likely to see $i$ as it is to see $j$.
	Thus, $i$ will be able to complete about twice as many operations as $j$,
	which is fair since $i$ is twice as fast as $j$.
	More generally, the ratio of the periods of times that the calls of help of any two processes $i$ and $j$ remain in $R_E$
	is approximately the ratio of their speeds $s_i / s_j$,
	so the ratio of the number of operations completed by $i$ and $j$ via the \DNB{2} help mechanism is about the ratio
	of their speeds $s_i / s_j$, as fairness requires.
	
	This light-weight help mechanism to achieve the 2-nonblocking property is not limited to queues.
	In fact, in Appendix~\ref{appendix-universal} we give a \emph{universal construction} for 2-nonblocking objects
	that essentially employs the same help mechanism:
	A process first makes a \emph{single} attempt to help one other process
	%	tries to help another process only \emph{once}
	(the ``altruistic'' phase) and then it repeatedly tries to perform its \emph{own} operation
	until it is done (the ``selfish'' phase).
	All the processes that need help compete by (over-)writing their call for help
	on the \emph{same} shared register.
	This guarantees the 2-nonblocking property,
	and it also promotes fairness as explained above.
	\section{Model sketch}
	We consider a standard distributed system where asynchronous processes that may fail by crashing
	communicate via shared registers and other objects,
	including synchronization objects such as \emph{compare\&swap (CAS)}~\cite{Herlihy91}.
	%Processes are asynchronous and they may fail by crashing.
	In such systems, shared objects can be used to \emph{implement} other shared objects
	such that: (1)~the implemented objects are \emph{linearizable}~\cite{HerlihyWing90}
	and (2) they satisfy some liveness requirement.
	In particular, we consider the $k$-nonblocking liveness requirement~\cite{BushkovGuerraoui15,Ben-DavidCHT16}:
	
	\begin{definition}
		$k$-nonblocking (\NB{k}):
		if a process invokes an operation
		and takes infinitely many
		steps without completing it,
		then at least $k$ processes complete infinitely many operations.
	\end{definition}
	
	We also introduce a liveness requirement, called \emph{differentiated $k$-nonblocking} or $\DNB{k}$ for short,
	that takes into consideration the fact that a shared object may have several operation \emph{types}
	(e.g., queues have enqueue and dequeue operations, atomic snapshot have write and scan operations, etc.).
	The differentiated $k$-nonblocking property of an object (or object implementation) is defined as follows:
	
	\begin{definition}
		Differentiated $k$-nonblocking (\DNB{k}):
		for every operation type $T$,
		if a process invokes an operation \emph{of type $T$}
		and takes infinitely many
		steps without completing it,
		then at least $k$ processes complete infinitely many operations \emph{of type $T$}.
	\end{definition}
	\section{A \DNB{2} algorithm for shared queue}

	\makeatletter
	\renewcommand{\ALG@name}{Fig.}
	\makeatother 
	
	%\MP{Numbers must be MANUALLY set.}
	
	\setcounter{algorithm}{0}

	\begin{figure}[t]
		\hrule
		\smallskip
		
		\begin{center}
			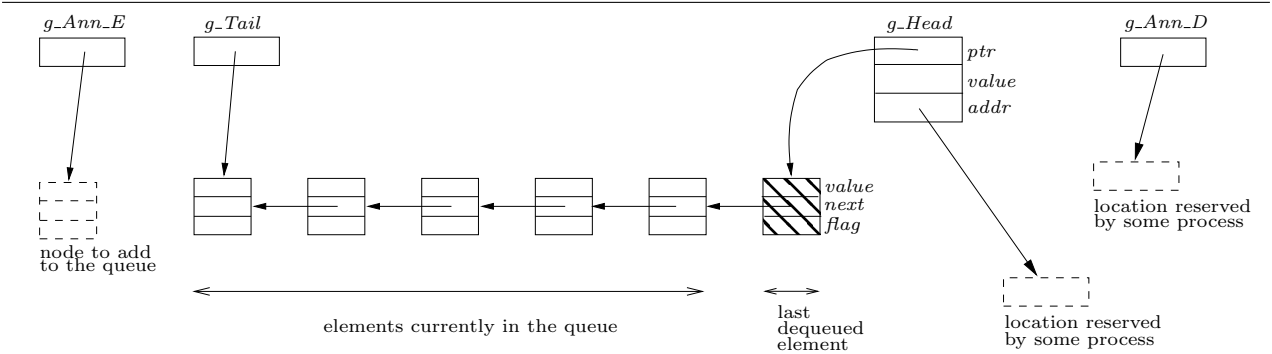
		\end{center}
		\caption{The data structures used by the algorithm}\label{data-structures-fig}
		\smallskip
		\hrule
	\end{figure}

	\subsection{Description of the algorithm}\label{algo-description}
	\textbf{Data structures.}\enspace
	We now describe the data structures used to represent the queue
	and to implement the helping mechanisms.
	(The algorithm incorporates two independent helping mechanisms
	to enforce the 2-nonblocking property,
	one for enqueuers and one for dequeuers.)
	The data structures are listed in Figure~\ref{alg0} and illustrated in Figure~\ref{data-structures-fig}.%\DMP{We may need to change this figure.}

	The queue consists of a linked list of nodes,
	containing the elements that have been enqueued but not yet dequeued,
	\bemph{as well as the last element dequeued},
	in the order in which these elements were enqueued.
	Each node is a record with the following three fields:
	\begin{compactitem}
		\item
		$\val$, the actual element enqueued;
		\item
		$\fNext$, a pointer to the next node in the queue
		(or $\mathnull$, if this node is the last one); and
		\item
		$\fFlag$, a Boolean set to true
		to indicate that the node has been inserted into the queue,
		used by the helping mechanism for enqueuers. 
	\end{compactitem}
	
	The tail of the queue (the end to which elements are added)
	is identified through a pointer, called $\gTail$,
	that points to the last node in the queue,
	unless a node is in the process of being added to the queue,
	in which case $\gTail$ could be temporarily pointing to the
	penultimate node in the queue.

	The head of the queue (the end from which elements are removed)
	is identified through a record, called $\gHead$,
	that contains the following information:
	\begin{compactitem}
		\item
		$\ptr$, a pointer to \bemph{the last node that was dequeued}.
		The queue is empty if and only if
		$\gHead.\ptr$ and $\gTail$ point to the same node.	
		If the queue is not empty,
		the first element in the queue ---
		i.e., the next element to be dequeued ---
		is in the node pointed to by $\gHead.\ptr\rightarrow\fNext$.
	\end{compactitem}

	\begin{algorithm}[h]
		\caption{Shared objects for the \DNB{2} queue algorithm}
		
		\label{alg0}
		\textbf{Structures:}
		\begin{algorithmic}
			\Statex \textbf{structure} $\textsc{node}$ \{
			\Statex \hspace{4mm} $\val$: a register containing an integer.
			\Statex \hspace{4mm} $\fNext$: a compare\&swap containing a pointer to a node.
			\Statex \hspace{4mm} $\fFlag$: a register containing a boolean.
			\}
			%\Statex
		\end{algorithmic}
		\vspace{10pt}
		\textbf{Shared Objects:} 
		\begin{algorithmic}
			\Statex $\gINode$: A node, initially with \{$\val = $ arbitrary, $\fNext = \mathnull$, and $\fFlag = 1$\}. 
			%		\Statex
			\Statex $\gIDQ$: A location containing an arbitrary non-$\mathnull$ value.
			%		\Statex
			\Statex $\gTail$: a compare\&swap containing a pointer to a node, 
			initially pointing to $\gINode$.
			%		\Statex
			\Statex $\gHead$: a compare\&swap with fields:
			\Statex \hspace{4mm} $\ptr$: a pointer to a node,
			initially pointing to $\gINode$.
			\Statex \hspace{4mm} $\val$: an integer, initially an arbitrary non-$\mathnull$ value.
			\Statex \hspace{4mm} $\add$: a pointer to a location that stores a dequeued value or $\mathnull$, initially pointing to $ \gIDQ $.
			%		\Statex
			\Statex $\gAnnE$: a register containing 
			a pointer to a node, initially $\gINode$.
			%		\Statex
			\Statex $\gAnnD$: a register containing a pointer to a location that stores a dequeued value or $\mathnull$, initially pointing to $ \gIDQ. $
			%		\Statex
		\end{algorithmic}
	\end{algorithm}
	\setcounter{algorithm}{0}
	\setcounter{figure}{1}
	\makeatletter
	\renewcommand{\ALG@name}{Algorithm}
	\makeatother

	\begin{compactitem}
		\item
		Two fields with information pertaining to
		the last dequeue operation,
		used by the helping mechanism for dequeuers:
		\begin{compactitem}
			\item
			$\val$, the last element dequeued.
			This is the same as $\gHead.\ptr\rightarrow\val$,
			except if the queue was empty when the last element was dequeued,
			in which case it is~$\emptyq$.
			\item
			$\add$, a pointer to a location reserved
			by the process that dequeued the last element;
			this location contains the element dequeued by that process or $\mathnull$.
		\end{compactitem}
	\end{compactitem}

	Initially, the linked list representing the queue
	contains a single node $\gINode$,
	with both $\gTail$ and $\gHead.\ptr$ pointing to it.
	The $\val$ field of $\gINode$ is arbitrary,
	$\fNext$ is $\mathnull$, and
	$\fFlag$ is~1.
	(This node can be thought of as representing
	a fictitious element that was enqueued and then dequeued
	before the algorithm starts.)
	The remaining two fields of $\gHead$, $\val$ and $\add$,
	are initialized to a non-$\mathnull$ value and
	a pointer to $\gIDQ$ (a location that contains an arbitrary non-$\mathnull$ value).
	
	To implement the helping mechanisms among enqueuers and among dequeuers,
	the algorithm uses two shared registers:
	\begin{compactitem}
		\item
		$\gAnnE$ contains a pointer to a node that some process
		wishes to add to the linked list of nodes
		representing the queue
		(initially a pointer to $\gINode$ whose $\fFlag$ field is~1,
		indicating that no help is needed); and
		\item
		$\gAnnD$ contains a pointer to the location reserved by a process $p$
		that wishes to perform a dequeue operation;
		it is intended for a helper to store the value dequeued for $p$
		(initially a pointer to $\gIDQ$, a location that contains
		a \emph{non}-$\mathnull$ value,
		indicating that no help is needed).
	\end{compactitem}

	\medskip\noindent
	\textbf{The enqueue operation.}\enspace
	Consider a process $p$ that wishes to enqueue an element $v$
	(see procedure $\NQ(v)$, lines~\ref{A1}-\ref{A19}).%\MP{Refer to Figures or Algorithms? Add some comments to algorithms, like   separating/commenting altruistic/selfish parts}
	Process $p$ must add a node containing $v$
	to the tail end of the linked list of nodes representing the queue.
	
	Roughly speaking, $\NQ(v)$ consists of two phases:
	In the \emph{altruistic phase} (lines~\ref{A2}--\ref{A5})
	$p$ tries \bemph{once} to help some process that
	has asked for help to enqueue an element.
	Whether it succeeds or fails, it then proceeds to
	the \emph{selfish phase} (lines~\ref{A12}--\ref{A14}),
	where $p$ keeps trying to enqueue its own element
	(a node it has prepared in lines~\ref{A8}--\ref{A11} to append to the linked list)
	until it succeeds to do so. Each time it tries and fails,
	$p$ asks for help by writing into the shared variable $\gAnnE$
	a pointer to the node it wants to append to the linked list.
	It is possible that the enqueue operation does not terminate
	because a process is stuck forever

	\begin{algorithm}[h]
		\caption{The \DNB{2} queue algorithm (code for process $p$)}
		\label{alg1}
		\begin{algorithmic}[1] 
			\Procedure{$ \NQ $}{$v$} ~~~\{ $v \in \Z$ \}\phantomsection\label{A1}
			\State $\lAnnE := \gAnnE.\textsc{read}()$\phantomsection\label{A2}
			\State $\trynq(\lAnnE)$\phantomsection\label{A5}
			\State $\lNode :=$ a new node\phantomsection\label{A8}
			\State $\lNode \rightarrow \val.\textsc{write}(v)$\phantomsection\label{A9}
			\State $\lNode \rightarrow \fNext.\textsc{write}(\mathnull)$\phantomsection\label{A10}
			\State $\lNode \rightarrow \fFlag.\textsc{write}(0)$\phantomsection\label{A11}
			\While{$ \trynq(\lNode)=\vFail $}\phantomsection\label{A12}
			\State $ \gAnnE.\textsc{write}(\lNode) $\phantomsection\label{A13}
			\EndWhile \phantomsection\label{A14}
			\State \textbf{return} $ \vDone $\phantomsection\label{A15}
			\EndProcedure\phantomsection\label{A19}
			\Statex

			\Procedure{$ \DQ $}{$ $}\phantomsection\label{B1}
			\State $\lAnnD := \gAnnD.\textsc{read}()$\phantomsection\label{B2}
			\If{$^*\lAnnD.\textsc{read}() = \mathnull$}\phantomsection\label{B5}
			\State $\trydq(\lAnnD)$\phantomsection\label{B6}
			\EndIf\phantomsection\label{B7}
			\State $\ladd:=$ a pointer to a new location\phantomsection\label{B10}
			\State $ ^*\ladd.\textsc{write}(\mathnull) $\phantomsection\label{B10.1}
			\Repeat\phantomsection\label{B11}
			\State $t := \trydq(\ladd)$\phantomsection\label{B12}
			\If{$ t=\vFail $}\phantomsection\label{B13}
			\State $\gAnnD := \ladd$\phantomsection\label{B14}
			\EndIf\phantomsection\label{B17}
			\Until {$t\ne\vFail$} \phantomsection\label{B18}
			\State \textbf{return} $t$\phantomsection\label{B19}
			\EndProcedure\phantomsection\label{B22}
			\rememberlines
		\end{algorithmic}
	\end{algorithm}

	\setcounter{algorithm}{0}
	
	\noindent in the selfish phase.
	As we will show, however, %(see Lemma~\ref{lemma:2NB-NQ})
	this can happen only if at least \emph{two} other processes succeed
	in enqueuing elements infinitely often.
	Note that an enqueuer never helps a dequeuer.
	
	Both $p$'s attempt to help another enqueuer and
	its attempt to enqueue its own element
	are carried out by the procedure $\trynq(r)$,
	discussed below,
	where $r$ is a pointer to the node to be appended to the list.
	$\trynq(r)$ returns $\TRUE$
	if it succeeds to append to the list the node pointed to by $r$;
	otherwise it returns $\FALSE$.
	The latter happens only if some other concurrent enqueue operation
	succeeds in appending a node to the list.
	
	We now describe the procedure $\trynq(r)$
	(lines~\ref{C1}--\ref{C18}).
	The procedure first creates
	a local copy $\lTail$ of the tail pointer $\gTail$, and
	a local copy $\afterTail$ of
	$\gTail\rightarrow\fNext$
	%	the $\fNext$ pointer of the node
	%	to which $\gTail$ points
	(lines~\ref{C2}--\ref{C3}).
	Ordinarily $\afterTail=\mathnull$,
	unless some call to $\trynq$ has appended a node
	to the linked list but has not (yet) updated the $\gTail$ pointer.
	\vspace{10pt}
	
	\begin{algorithm}[h]
		\caption{The $\trynq$ and $\trydq$ procedures}\label{alg2}
		\begin{algorithmic}[1]\resumenumbering
			\Procedure{\trynq}{$\lNode$}\phantomsection\label{C1}
			\State $\lTail := \gTail.\textsc{read}()$\phantomsection\label{C2}
			\State $\lnext := \lTail \rightarrow \fNext.\textsc{read}()$\phantomsection\label{C3}
			\If{$\lNode \rightarrow \fFlag.\textsc{read}() = 1$}\phantomsection\label{C4}
			\State $\lTail := \gTail.\textsc{read}()$\phantomsection\label{C2repeat}
			\State $\lnext := \lTail \rightarrow \fNext.\textsc{read}()$\phantomsection\label{C3repeat}
			\If{$ \lnext\ne\mathnull $}\phantomsection\label{C4.1}
			\State $ \lnext\rightarrow\fFlag.\textsc{write}(1) $\phantomsection\label{C5}
			\State $ \gTail.\CAS(\lTail,\lnext) $\phantomsection\label{C5.1}
			\EndIf\phantomsection\label{C5.2}
			\State \textbf{return} $ \vDone $\phantomsection\label{C5.3}
			\EndIf\phantomsection\label{C6}
			\If{$\lnext \ne \mathnull$} \phantomsection\label{C7}
			\State $\lnext \rightarrow \fFlag.\textsc{write}(1)$\phantomsection\label{C14}
			\State $\gTail.\CAS(\lTail,\lnext)$\phantomsection\label{C15}
			\Else
			\If{$\lTail \rightarrow \fNext.\CAS(\mathnull,\lNode)$}\phantomsection\label{C8}
			\State $\lNode \rightarrow \fFlag.\textsc{write}(1)$\phantomsection\label{C9}
			\State $\gTail.\CAS(\lTail,\lNode)$\phantomsection\label{C10}
			\State \textbf{return} $\vDone$ \phantomsection\label{C11}
			\EndIf\phantomsection\label{C12}
			\EndIf\phantomsection\label{C17}
			\State \textbf{return} $\vFail$ \phantomsection\label{C16}
			\EndProcedure\phantomsection\label{C18}
			\rememberlines
			\Statex

			\Procedure{$ \trydq $}{$\ladd$}\phantomsection\label{D1}
			\State $\lHead := \gHead.\textsc{read}()$\phantomsection\label{D2}
			\State $\lTail := \gTail.\textsc{read}()$\phantomsection\label{D3}
			\State $ ^*(\lHead.\add).\textsc{write}(\lHead.\val) $\phantomsection\label{D6}
			\If{$^*\ladd.\textsc{read}() \neq \mathnull$}\phantomsection\label{D9}
			\State \textbf{return} $^*\ladd.\textsc{read}()$\phantomsection\label{D10}
			\EndIf\phantomsection\label{D11}
			\If{$\lHead.\ptr = \lTail$}\phantomsection\label{D13}
			\If{$\gHead.\CAS(\lHead, \langle \lHead.\ptr, \emptyq, \ladd \rangle)$}\phantomsection\label{D15}
			\State \textbf{return} $\emptyq$\phantomsection\label{D16}
			\EndIf\phantomsection\label{D17}
			\Else\phantomsection\label{D22}
			\State $\lnext := \lHead.\ptr \rightarrow \fNext.\textsc{read}()$\phantomsection\label{D4}
			\State $v := \lnext \rightarrow \val.\textsc{read}()$\phantomsection\label{D23}
			\If{$\gHead.\CAS(\lHead, \langle \lnext, v, \ladd \rangle)$}\phantomsection\label{D24}
			\State \textbf{return} $v$\phantomsection\label{D25}
			\EndIf\phantomsection\label{D26}
			\EndIf\phantomsection\label{D27}
			\State \textbf{return} $\vFail$\phantomsection\label{D29}
			\EndProcedure \phantomsection\label{D30}

		\end{algorithmic}
		
	\end{algorithm}

	$\trynq(r)$ checks whether the node to which $r$ points
	has $\fFlag=1$,
	indicating that the node has already been threaded to the list (line~\ref{C4}).
	If so,
	$\trynq$ refreshes $\lTail$ and $\afterTail$
	(lines~\ref{C2repeat}--\ref{C3repeat}).
	It then ensures that the node pointed to by $\afterTail$, if any,
	is fully incorporated into the data structures:
	its $\fFlag$ is set to~1 and $\gTail$ points to it,\footnote{Pointer $\gTail$ is updated using a CAS operation
		to ensure that a late application does not over-write
		an earlier update.}
	and returns $\TRUE$
	(lines~\ref{C4.1}--\ref{C5.3}).
	%this call to $\trynq$ returns $\TRUE$,
	%after ensuring that the node pointed to by $\afterTail$,
	%if any, is fully incorporated into the data structures:
	%its $\fFlag$ is set to~1 and $\gTail$ points to it
	%(lines~\ref{C4.1}--\ref{C5.1}).\RMP{Mention the CAS used here?}
	
	If the node to which $r$ points does not have $\fFlag=1$,
	then the $\trynq(r)$ procedure tests whether $\afterTail=\mathnull$ (line~\ref{C7}).
	If not,
	a concurrent $\NQ$ operation succeeded in appending its node to the list.
	So $\trynq(r)$ ensures that the appended node
	is fully incorporated into the data structures:
	its $\fFlag$ is set to~1 and $\gTail$ points to it (lines~\ref{C14}-\ref{C15});
	and returns $\FALSE$ (line~\ref{C16}).
	If, on the other hand, $\trynq(r)$ found that $\afterTail=\mathnull$,
	it tries to append to the list the node pointed to by $r$
	by applying a $\CAS(\mathnull,r)$ operation on 
	$\lTail\rightarrow\fNext$
	%	the $\fNext$ field of the node pointed to by $\lTail$
	(line~\ref{C8}).
	If this $\CAS$ is successful, $\trynq$
	incorporates into the list the new node pointed to by $r$
	by setting its $\fFlag$ to~1 and ensuring that $\gTail$ points to it,
	and then returns $\TRUE$ (lines \ref{C9}--\ref{C11}).
	If the $\CAS$ is unsuccessful,
	then another enqueue operation succeeded in appending to the linked list
	a node other than $r$.
	Thus, $\trynq$ returns $\FALSE$ (line~\ref{C16}).
	
	\medskip\noindent
	\textbf{The dequeue operation.}\enspace
	Consider a process $p$ that wishes to perform a dequeue operation
	(see procedure $\DQ$, lines \ref{B1}--\ref{B22}).
	Process $p$ must return the first element in the queue
	(or $\emptyq$, if the queue is empty)
	after updating the information in $\gHead$
	to reflect the removal of that element.
	
	Similar to $\NQ$, the $\DQ$ operation consists of
	an altruistic and a selfish phase:
	In the altruistic phase (lines~\ref{B2}--\ref{B7})
	$p$ tries \bemph{once} to help some process that
	has asked for help to dequeue an element.
	Whether it succeeds or fails,
	$p$ then proceeds to the selfish phase (lines~\ref{B10}--\ref{B18}),
	where it keeps trying to dequeue an element for itself
	until it succeeds to do so;
	each time it tries and fails,
	$p$ asks for help by writing into the shared variable $\gAnnD$
	a pointer to a location it has reserved (line~\ref{B10}),
	in which the helper process will place the dequeued element
	for $p$ to retrieve.
	This location is initialized to the value $\mathnull$ (line~\ref{B10.1}),
	which indicates that $p$ (the process requesting help)
	has not yet received help;
	when this location contains a non-$\mathnull$ value,
	$p$'s dequeue operation has been helped, and
	$p$ can retrieve the dequeued element from that location.
	So, the altruistic loop starts with $p$ making a local copy $\lAnnD$
	of the $\gAnnD$ shared variable used to ask for help (line~\ref{B2}).
	If the location to which $\lAnnD$ points contains $\mathnull$,
	$p$ tries to help a dequeuer (lines~\ref{B5}-\ref{B7}).
	
	It is possible that the dequeue operation does not terminate:
	A process may be stuck forever in the selfish phase.
	As we will show %(see Lemma~\ref{lemma:2NB-DQ})
	this can happen only if at least \emph{two} other processes succeed
	in dequeuing elements infinitely often.
	Note that a dequeuer never helps an enqueuer.
	
	Both $p$'s attempt to help another dequeuer and
	its attempt to dequeue its own element
	are carried out by the procedure $\trydq(r)$,
	discussed below,
	where $r$ is a pointer to a location reserved by the process
	on behalf of which this attempt to dequeue is being made.\footnote{The
		process on behalf of which $p$'s call to $\trydq(r)$
		is trying to dequeue an element
		is a process other than $p$,
		if this call is made in the altruistic phase;
		or $p$ itself,
		if this call is made in the selfish phase.}
	If it is successful, $\trydq(r)$ returns the element dequeued
	(or the special value $\emptyq$, if the queue is empty);
	otherwise it returns $\FALSE$.
	The latter happens only if some other concurrent dequeue operation
	succeeds in dequeuing an element.
	
	We now describe the procedure $\trydq(\ladd)$
	(lines~\ref{D1}--\ref{D30}). %\RMP{I use $r$ instead of $\ladd$.}
	This procedure first makes a local copy $\lHead$ of the $\gHead$ record
	and a local copy $\lTail$ of the $\gTail$ pointer
	(lines~\ref{D2}--\ref{D3}).
	Recall that the fields $\val$ and $\add$ of $\gHead$
	contain information about the last dequeue operation performed.
	As this information is about to be over-written,
	and the value dequeued by that operation (namely, $\gHead.\val$)
	may be needed to help the process that
	%performed
	issued
	that dequeue operation
	(namely, the process that reserved the location pointed to
	by $\gHead.\add$),
	the dequeued element is copied to that location (line~\ref{D6}).
	
	Once the information in $\gHead$ has been safely stored,
	$\trydq(\ladd)$ examines the value in the location pointed to by $\ladd$.
	If this is non-$\mathnull$,
	the element dequeued for this operation
	(or $\emptyq$, if the queue was empty)
	has already been stored there by a helper.
	Thus, in this case,
	$\trydq(\ladd)$ merely returns
	the content of the location pointed to by $\ladd$ (lines~\ref{D9}--\ref{D11}).
	Otherwise, $\trydq(\ladd)$
	tries to dequeue an element
	by attempting to update $\gHead$.
	This is done by means of a CAS operation that changes $\gHead$
	if it has not been changed since $\trydq(\ladd)$ made a local copy of
	that variable in $\lHead$.
	The new information written into $\gHead$ depends on whether
	$\trydq(\ladd)$ found the queue to be empty (lines~\ref{D13}--\ref{D17})
	or not (lines~\ref{D22}--\ref{D27}).
	If the CAS operation succeeds, $\trydq(\ladd)$ returns
	the dequeued element (or $\emptyq$, if the queue is empty).
	%	
	%If the queue is empty (i.e., $\lTail=\lHead.\ptr$) and
	%	the CAS operation succeeds,
	%	this dequeue operation should return $\emptyq$ and 
	%	$\gHead$ should be updated as follows:\RMP{Perhaps this paragraph is unnecessary.}
	%\begin{compactitem}
	%	\item
	%	$\gHead.\ptr$ should remain unchanged since there is no
	%	node to remove from the list;
	%	\item
	%	$\gHead.\val$ should be set to the value returned
	%	by this dequeue operation, namely $\emptyq$;
	%	\item
	%	$\gHead.\add$ should point to the location pointed to by $r$
	%	(i.e., the location in which the dequeued element should be
	%	stored by a helping dequeuer).
	%\end{compactitem}
	%
	%If the queue is not empty and
	%	the CAS operation succeeds,
	%	this dequeue operation should return
	%	the element in the node after the last deleted node,
	%	i.e., $(\lHead.\ptr\rightarrow\fNext)\rightarrow\val$ and
	%	$\gHead$ should be updated as follows (line~72):\RMP{Perhaps this paragraph is unnecessary.}
	%\begin{compactitem}
	%	\item
	%	$\gHead.\ptr$ should point to the next node on the list,
	%	i.e., to $\lHead.\ptr\rightarrow\fNext$;
	%	\item
	%	$\gHead.\val$ should be set to the element that will be returned
	%	by this dequeue operation, which is the element 
	%	in the next node of the list,
	%	i.e., $(\lHead.\ptr\rightarrow\fNext)\rightarrow\val$; and
	%	\item
	%	$\gHead.\add$ should point to the location pointed to by $r$.
	%\end{compactitem}
	%
	If the CAS operation fails,
	then some concurrent dequeuer succeeded
	in effecting its dequeue operation,
	and this call to $\trydq$ returns $\FALSE$ (line~\ref{D29}).
	
	\newpage
	\subsection{Proof of linearizability}
	%\noindent\textbf{Proof of Linearizability}
	
	Let $H$ be any history of the algorithm.
	Then we construct a completion $H'$ of $H$ as follows:
	\begin{itemize}
		\item For each incomplete $\NQ$ operation in $H$, 
		if it allocates a new node on line~\ref{A8} 
		and this node has been pointed to by the $\gTail$ pointer,
		then it is completed in $H'$ by returning $\vDone$;
		otherwise, it is removed from $H'$.
		
		\item For each incomplete $\DQ$ operation in $H$,
		if it allocates a new location on line~\ref{B10}
		and this location has been pointed to by $\gHead.\add$,
		then it is completed in $H'$ by returning the value of $\gHead.\val$
		at the earliest time when $\gHead.\add$ 
		points to the new location it allocated;
		otherwise, it is removed from $H'$. 
	\end{itemize}
	
	We then construct a linearization $L$ of $H'$ as follows:
	\begin{itemize}
		\item Each $\NQ$ operation in $H'$ is linearized
		when $\gTail$ first points to the new node it allocates on line~\ref{A8}.
		
		\item Each $\DQ$ operation is linearized
		when $\gHead.\add$ first points to the new location 
		it allocates in line~\ref{B10}
		if the return value of the $\DQ$ is not $\emptyq$;
		otherwise it is linearized
		when the first process that sets $\gHead.\add$ 
		to point to the new location
		last reads $\gTail$ in line~\ref{D3}.
	\end{itemize}
	
	The proof that this linearization is correct
	with respect to the specification of a queue
	can be found in \autoref{asec:appendix-algo-proof}. 
	
	\subsection{Proof of liveness: the \DNB{2} property}
	
	\begin{definition}\label{NB-def}
		%\noindent\textbf{Definition 2:}
		A shared queue is called \emph{non-blocking for $\NQ$ operations},
		if it satisfies the following property:
		if a process invokes an $\NQ$ operation and takes infinitely many steps 
		without completing it, then
		there must be another process completing
		infinitely many $\NQ$ operations.
		\textit{Nonblocking for $\DQ$ operations} is defined similarly.
	\end{definition}
	
	\begin{definition}\label{2NB-def}
		%\noindent\textbf{Definition 3:}
		A shared queue is called \emph{2-nonblocking for $\NQ$ operations},
		if it satisfies the following property:
		if a process invokes an $\NQ$ operation and takes infinitely many steps 
		without completing it, then
		there must be at least two other processes completing
		infinitely many $\NQ$ operations.
		\textit{2-nonblocking for $\DQ$ operations} is defined similarly.
	\end{definition}
	
	We first show that the algorithm is non-blocking, and
	then use this to prove that it is also 2-nonblocking.
	
	\begin{theorem}\label{NB-proof}
		%\noindent\textbf{Theorem 3:}
		The algorithm is non-blocking for both $\NQ$ and $\DQ$ operations.
	\end{theorem}
	
	\begin{proof}
		First consider $\NQ$ operations.
		Suppose, for contradiction, that
		some process $p$ takes infinitely many steps for an $\NQ$ operation
		but no other process completes an infinite number of $\NQ$ operations.
		Thus, eventually
		$p$ makes continuously calls to $\trynq$,
		all of which return $\vFail$;
		and every other process
		ceases to execute any operations, or
		executes infinitely many $\DQ$ operations, or
		is stuck in an $\NQ$ or $\DQ$ operation forever.
		Note that $\trynq$ returns $\vFail$ either
		because the condition in line~\ref{C7} is true, or
		because the CAS in line~\ref{C8} fails.
		We now prove that each of these can happen only a finite number of times,
		contradicting that $p$ is stuck forever in an $\NQ$ operation.
		
		\begin{case}
			\item
			Process $p$ finds the condition in line~\ref{C7} to be true
			infinitely many times.
			Whenever $p$ finds the condition in line~\ref{C7} to be true
			it performs a CAS on $\gTail$ to move it forward 
			through the list of nodes, and
			this CAS fails only if another process has already moved the tail forward.
			This means that infinitely many nodes are added to the list,
			i.e., infinitely many $\NQ$ operations complete,
			contrary to our supposition.
			
			\item
			The CAS in line~\ref{C8} fails infinitely many times.
			Whenever this CAS fails, some node is added to the list.
			This implies that an infinite number of $\NQ$ operations complete,
			contrary to our supposition.
		\end{case}
		
		Next we prove that the algorithm is also non-blocking for $\DQ$ operations.
		By a similar argument, eventually some process
		$p$ makes continuously calls to $\trydq$,
		all of which return $\vFail$;
		and every other process
		ceases to execute any operations, or
		executes infinitely many $\NQ$ operations, or
		is stuck in an $\NQ$ or $\DQ$ operation forever.
		Note that a call to $\trydq$ returns $\vFail$ only
		if the CAS is line~\ref{D15} or~\ref{D24} fails.
		Thus, $p$ performs
		an infinite number of unsuccessful CAS operations on $\gHead$.
		This means that there are
		infinitely many successful CAS operations on $\gHead$.
		Therefore, an infinite number of $\DQ$ operations complete,
		contrary to our assumption.
	\end{proof}
	\newpage
	
	\begin{theorem}\label{2NBNQ-thm}
		%noindent\textbf{Theorem 4:}
		The algorithm is 2-nonblocking for $\NQ$ operations.
	\end{theorem}
	
	\begin{proof}
		Suppose, for contradiction, that the algorithm is not 2-nonblocking 
		for $\NQ$ operations.
		Thus there is an execution where a process takes infinitely many steps 
		but is stuck in an $\NQ$ operation forever, 
		and during which fewer than two processes complete 
		infinitely many $\NQ$ operations. 
		Since we have proved the algorithm is non-blocking for $\NQ$ operations, 
		the only possibility is that one process $p$ 
		completes infinitely many $\NQ$ operations. 
		Since other processes complete $\NQ$ operations 
		for only finitely many times,
		eventually (from some time $T$ onwards) 
		only $p$ can complete $\NQ$ operations. 
		We should notice that our algorithm has a help mechanism: 
		by \textit{only $p$ can complete operations} 
		we mean that only the requests of $p$ are fulfilled. 
		This does not necessarily means that process $p$ wins 
		every competition of setting $\gTail.\fNext$. 
		Other process can also win that competition, 
		but they can only succeed when they are helping $p$. 
		When process $p$ fails one such competition, 
		it will put its node in $\gAnnE$.
		We show this can only happen finitely many times.
		\newline
		
		\begin{claim}\label{2NBNQ-claim1}
			%\noindent\textbf{Claim 1:}
			After time $T$, process $p$ can fail on line~\ref{C7} or \ref{C8} 
			for at most $ 2n $ times.
		\end{claim}
		
		\begin{proof}[Proof of Claim~\ref{2NBNQ-claim1}]
			If process $ p $ fails one such competition 
			or finds that tail is not pointing to the last node, 
			then at least one other process has succeeded in linking 
			one node to the queue. 
			However, we have assumed that after time $T$,
			only process $p$ can return from an $\NQ$ operation. 
			Every other process either gets stuck in a loop forever, 
			or is able to carry on $\DQ$ operation but never enqueues. 
			They can succeed the CAS of line~\ref{C8} only when 
			helping process $p$, or when helping itself 
			but crashing before returning. 
			Either case only happens at most once for one process, 
			so $p$ can fail no more than $2n$ times.
		\end{proof}
		
		\begin{claim}\label{2NBNQ-claim2}
			%\noindent\textbf{Claim 2:}
			Eventually the value in register $ \gAnnE $ is never $p$.
		\end{claim}
		
		\begin{proof}[Proof of Claim~\ref{2NBNQ-claim2}]
			From Claim~\ref{2NBNQ-claim1}, we know that process $p$ 
			only sets $ \gAnnE $ to its node(s) for finitely many times, 
			because setting $ \gAnnE $ to a node of $p$ means 
			$p$ failed one competition on line~\ref{C7} or \ref{C8}, 
			which happens no more than $ 2n $ times. 
			However, by our assumption, some other process 
			is taking infinitely many steps in an $\NQ$ operation and never returns. 
			Thus it will set $ \gAnnE $ to its node infinitely many times. 
			As a result, eventually the register $ \gAnnE $ 
			points to a node that does not belong to $p$ 
			and is never rewritten by $p$.
		\end{proof}
		
		Now return to the proof of Theorem~\ref{2NBNQ-thm}.
		By Claim~\ref{2NBNQ-claim2}, eventually $ \gAnnE $ 
		never points to a node of $p$. 
		Since process $p$ completes infinitely many $\NQ$ operations, 
		after this time it will invoke an $\NQ$ operation 
		and see $\gAnnE$ contains a node from another process. 
		It will then work for that process on line~\ref{A5}.
		By Claim~\ref{2NBNQ-claim1}, eventually process $p$ never fails 
		the competitions on line~\ref{C7} or \ref{C8}, 
		so it succeeds in helping one other process. 
		Eventually the node in register $\gAnnE$ is from a process 
		that takes infinitely many steps, 
		(because if a node is written only finitely many times, 
		it will eventually be covered by those writing infinitely many times),
		and that process will complete its $\NQ$ operation 
		after being helped by process $p$, a contradiction. 
		In conclusion, the assumption that only process $p$ 
		can complete operations after time $ T $ is wrong, 
		so the algorithm is 2-nonblocking for the $\NQ$ function.
	\end{proof}
	
	\begin{theorem}\label{2NBDQ-thm}
		%\noindent\textbf{Theorem 5:}
		The algorithm is 2-nonblocking for the $\DQ$ operations.
	\end{theorem}
	
	\begin{proof}
		The proof is similar to that of Theorem~\ref{2NBNQ-thm}. 
		Eventually (from some time $ T $ onwards) only a process $ p $ 
		can complete infinitely many $\DQ$ operations, 
		and every other process is either stuck in a loop, 
		or never dequeues. 
		After time $ T $, process $ p $ can fail the CAS operations
		on lines~\ref{D15} or \ref{D24} for at most $ 2n $ times. 
		This is because that when process $ p $ fails one such CAS operation, 
		at least one other process succeeds to dequeue a node (or $ \emptyq $). 
		These events can happen at most twice for one process, 
		so the total number of times is no more than $ 2n$. 
		Eventually $ \gAnnD $ never stores a pointer allocated by $ p$, 
		because $ p $ never writes into it, 
		whereas some other process writes into it infinitely many times. 
		As a result, process $ p $ works for that process on line~\ref{B6}
		and successfully dequeues a node (or $\emptyq$) for it.
	\end{proof}
	
	From Theorems~\ref{2NBNQ-thm} and \ref{2NBDQ-thm}, we have:
	
	\begin{corollary}\label{2NBD-thm}
		%\noindent\textbf{Theorem 5:}
		The algorithm is \DNB{2}.
	\end{corollary}
	
	\newpage
	\section{Experimental results}
	\subsection{Experimental setting}
	We implemented two shared queue algorithms:
	our \DNB{2} algorithm (Algorithm 1), and
	the non-blocking MS algorithm as described in~\cite{MichaelScott96}.
	The two algorithms were implemented ``as is'', without any optimization:
	our goal was
	to get a rough idea of the $\DNB{2}$ algorithm's potential compared
	to a commonly used non-blocking algorithm.
	In particular, we wanted to explore the tradeoff between the cost and benefits
	of
	the non-blocking property of the MS algorithm (which does not need any help mechanism)
	compared to
	the stronger 2+2 non-blocking property of \DNB{2}
	(which uses a light-weight help mechanism).
	
	We implemented both algorithms in Java, and executed them on an
	Intel i7-9750H CPU with 6 cores (12 threads) and 16 GBs of RAM.
	These are only preliminary results as we intend to re-run more extensive experiments
	on a dedicated multiprocessor cluster in the future.
	The experimental results that we obtained so far indicate that,
	compared to the MS algorithm, the fairness of the $\DNB{2}$ algorithm is much higher
	across a wide range of process speeds,
	albeit at the cost of a marginal reduction in throughput.
	
	To evaluate the fairness of the two algorithms under various process speeds,
	we had to control the speed of individual processes.
	To control the speed of a process $p$, we added a delay immediately after every \emph{shared memory access}
	(read, write or CAS) by $p$;\footnote{We did not add delays after local steps,
		e.g., after reading a local variable,
		because the time to perform a local step is negligible compared to the time
		it takes to access shared memory.}
	In our experiments, these delays follow an exponential distribution $\ex(\mu)$,
	where $\mu$ is the average delay.
	So the number of shared-memory steps that $p$ executes by some time $t$ follows a
	Poisson distribution with parameter $\lambda = 1 / \mu$,
	and the simulated speed of a process $p$ is (proportional) to $1/\mu$.\footnote{This is also because
		the smallest delay $\mu$ that we chose in our experiments, 1ms, is orders of magnitude greater
		than the time that a process actually takes to execute any line of code.}
	Thus, in our experiments we controlled the speed of each process by setting its corresponding average delay $\mu$.
	
	%As we will show later, this model is equivalent with the Stochastic Scheduler model defined in [SCU].

	\subsection{Experiments with two enqueuers and two dequeuers}
	
	One of our experiments considered a system
	with only two enqueuers and two dequeuers %that repeatedly apply their operations,
	where we slowed one of the two enqueuers and one of the two dequeuers  by a factor of $k$,
	for each $k$ in the range $2 \le k < 20$.
	To do so, we set the delay parameter $\mu$
	of the slower enqueuer [dequeuer] to be $k$ times
	as large as the delay parameter $\mu$ of the faster enqueuer [dequeuer].
	
	\begin{figure}[h]
		\centering
		\begin{minipage}[t]{0.5\linewidth}
			\centering
			\includegraphics[width=7cm,height=5.5cm]{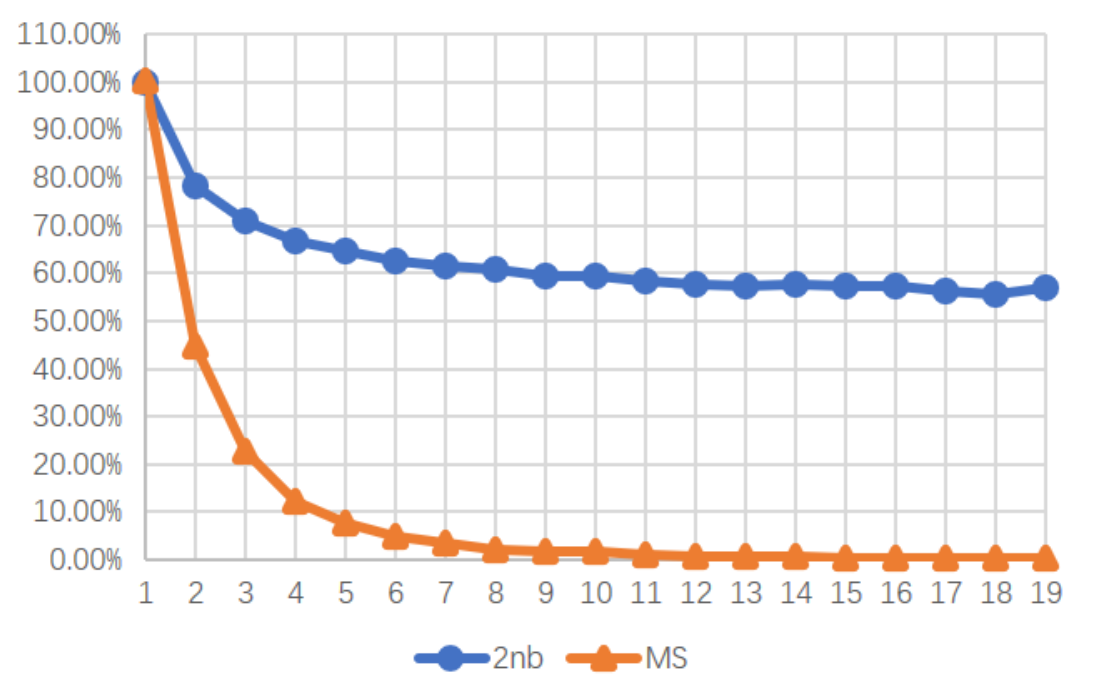}
			%			\caption{Fairness: slower enqueuer}
			\caption{\small{2 enqueuers + 2 dequeuers \\\hspace*{8mm}$y$-axis: \% of fair share obtained by slow enqueuer \\
					\hspace*{8mm}$x$-axis: slowdown factor of slow enqueuer}}
			\label{2+2nq}
		\end{minipage}%
		\begin{minipage}[t]{0.5\linewidth}
			\hspace{2mm}
			\includegraphics[width=7cm,height=5.5cm]{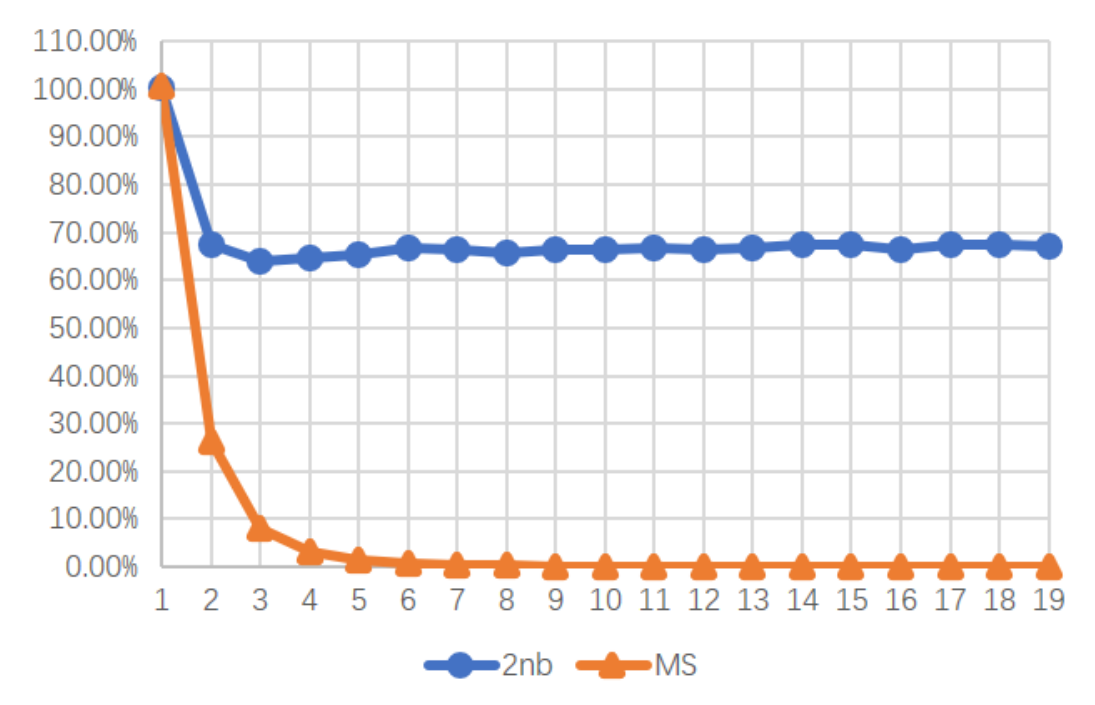}
			%			\caption{Fairness: slower dequeuer}
			\caption{\small{2 enqueuers + 2 dequeuers \\\hspace*{8mm}$y$-axis: \% of fair share obtained by slow dequeuer \\
					\hspace*{8mm}$x$-axis: slowdown factor of slow dequeuer}}
			\label{2+2dq}
		\end{minipage}%
	\end{figure}
	
	With the MS algorithm, we observed that the slower enqueuer and the slower dequeuer
	completed much less than their fair shares of operations as $k$ increased (see Figures~\ref{2+2nq} and~\ref{2+2dq}).
	And when we reached $k=11$, the slow dequeuer was not able to complete \emph{any} dequeue operation,
	while the other dequeuer completed about 49,000 operations;
	and the slow enqueuer completed only 43 operations,
	while the other enqueuer completed almost 59,000 operations.

	In contrast, the fairness of the $\DNB{2}$ algorithm remained reasonable (55\% or more)
	for both enqueuers and dequeuers,
	and for all the values of $k$ in the range $2 \le k < 20$.
	
	The throughput of the $\DNB{2}$ algorithm, however,
	was lower than that of the MS algorithm:
	for $k$ in the range $2 \le k < 20$, 
	the enqueuers' throughput of the $\DNB{2}$ algorithm was about 67\% to 62\%
	of the throughput of the MS algorithm, while the dequeuers' throughput was about  88\% to 76\% of the MS algorithm.
	
	\subsection{Experiments with eight enqueuers and eight dequeuers}
	
	The experimental results that we observed for a system with two enqueuers and two dequeuers
	were very encouraging,
	but since a $\DNB{2}$ queue is, by definition, wait-free for this special case,
	one may wonder whether the good fairness behaviour of our algorithm
	would continue to hold in systems with more processes.
	To check this, we considered a system with 8 enqueuers and 8 dequeuers.
	
	As a baseline, we first executed the two algorithms in a setting ${\mathcal{S}}_0$
	where all 16 processes have exactly the same speed (every process has $k=1$).
	These runs confirmed that, as expected, both algorithms are indeed fair in this case:
	every process got very close to 100\% of its fair share.\footnote{In each of the two groups (of 8 enqueuers and 8 dequeuers)
		each process completed about 12.5\% of the total number of operations completed by its group.}
	
	We then considered what happens when, in each group, processes have different speeds.
	To do so, we experimented with two different settings:
	
	\begin{itemize}
		\item a setting ${\mathcal{S}}_1$ where the differences in speeds are relatively small;
		specifically, process $i$, for $1 \le i \le 8$,
		is slowed down by a factor of $i$ (e.g., the speeds of processes 2, 3 and 4, are 1/2, 1/3 and 1/4 
		of the speed of process 1, respectively).
		\item a setting ${\mathcal{S}}_2$ where the differences in speeds are large;
		specifically, process $i$, for $1 \le i \le 8$,
		is slowed down by a factor of $2^{i-1}$ (e.g., the speeds of processes 2, 3 and 4, are 1/2, 1/4 and 1/8
		of the speed of process 1, respectively).
	\end{itemize}
	
	\begin{figure}[h]
		\centering
		\begin{minipage}[t]{0.5\linewidth}
			\centering
			\includegraphics[width=7cm,height=5.5cm]{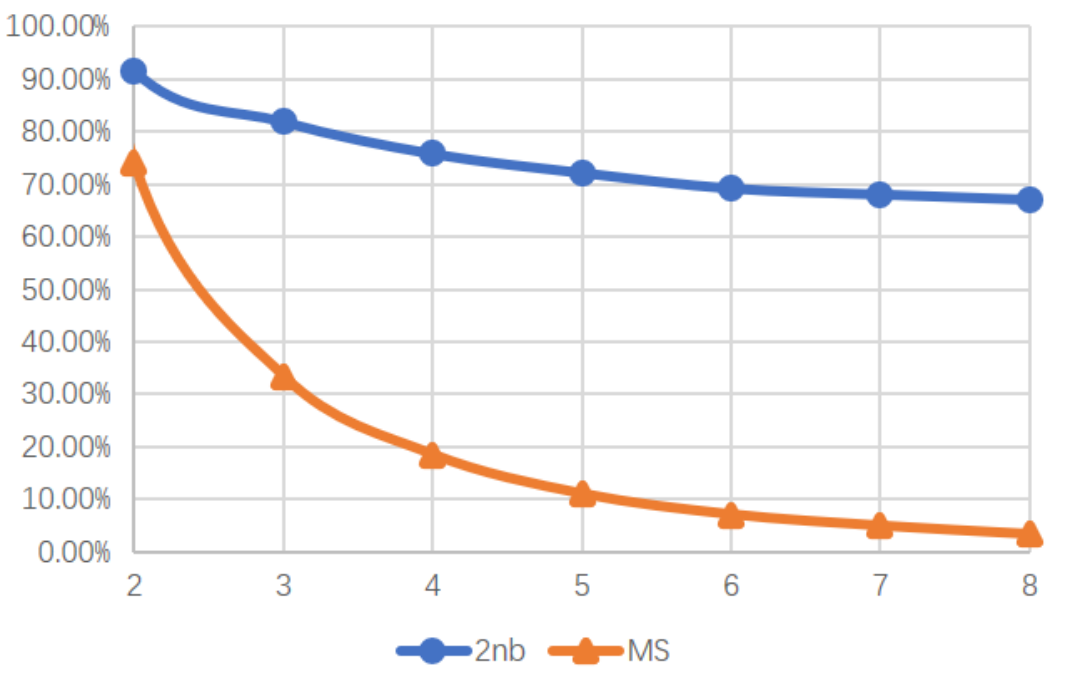}
			%		\caption{%${\mathcal{S}}_1$: 
			%		Fairness: 8 enqueuers with slowdown factor $ i $}
			\caption{\small{8 enqueuers + 8 dequeuers \\\hspace*{8mm}$y$-axis: \% of fair share obtained by enqueuer $i$ \\
					\hspace*{8mm}$x$-axis: enqueuer $i$, which is slowed by factor of $i$}}
			\label{1234nq}
		\end{minipage}%
		\begin{minipage}[t]{0.5\linewidth}
			\hspace{2mm}
			\includegraphics[width=7cm,height=5.5cm]{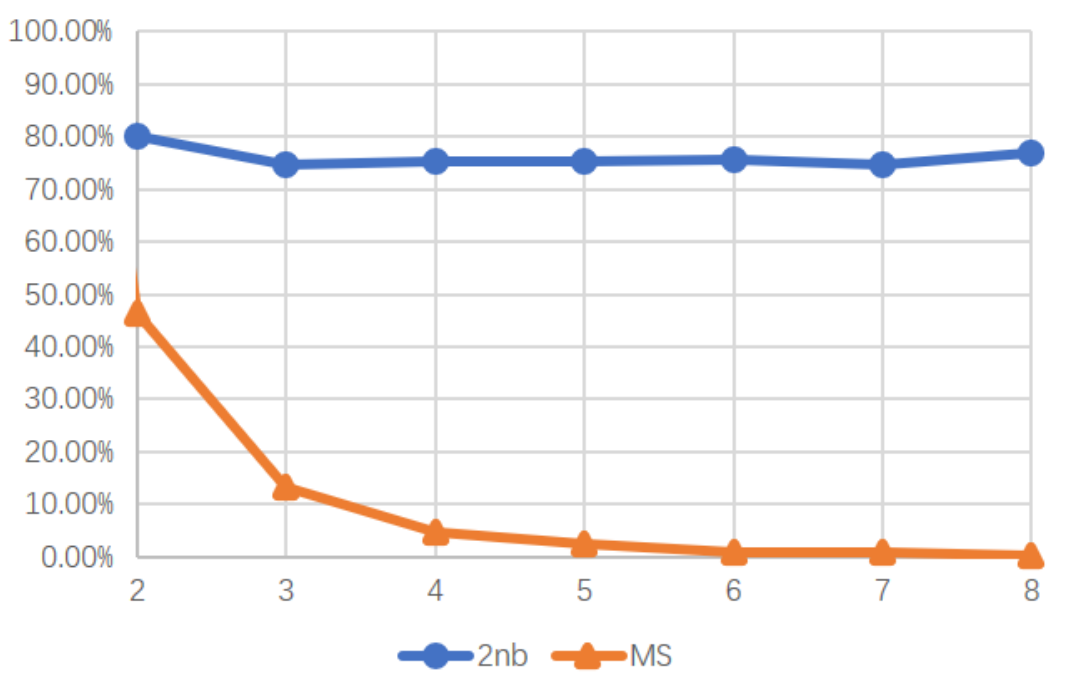}
			%		\caption{%${\mathcal{S}}_1$:
			%		Fairness: 8 dequeuers with slowdown factor $ i $}
			\caption{\small{8 enqueuers + 8 dequeuers \\\hspace*{8mm}$y$-axis: \% of fair share obtained by dequeuer $i$ \\
					\hspace*{8mm}$x$-axis: dequeuer $i$, which is slowed by factor of $i$}}
			\label{1234dq}
		\end{minipage}%
	\end{figure}
	
	\begin{figure}[h]
		\centering
		\begin{minipage}[t]{0.5\linewidth}
			\centering
			\includegraphics[width=7cm,height=5.5cm]{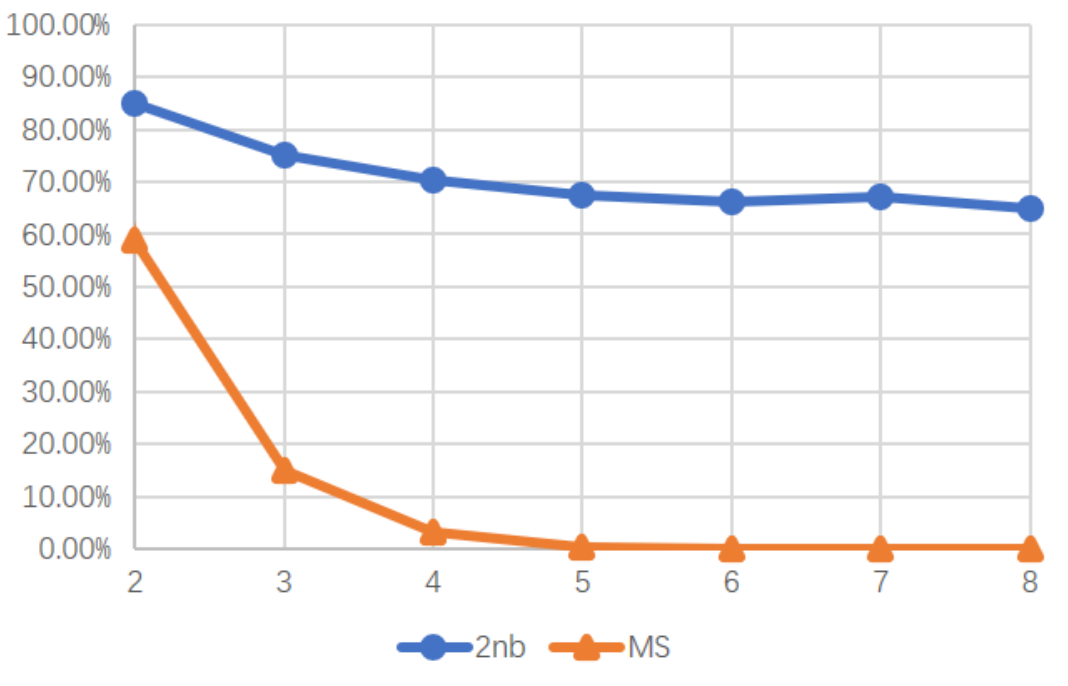}
			%		\caption{%${\mathcal{S}}_2$:
			%		Fairness: 8 enqueuers with slowdown factor $ 2^{i-1} $}
			\caption{\small{8 enqueuers + 8 dequeuers \\\hspace*{8mm}$y$-axis: \% of fair share obtained by enqueuer $i$ \\
					\hspace*{8mm}$x$-axis: enqueuer $i$, which is slowed by factor of $2^{i-1}$}}
			\label{1248nq}
		\end{minipage}%
		\begin{minipage}[t]{0.5\linewidth}
			\hspace{2mm}
			\includegraphics[width=7cm,height=5.5cm]{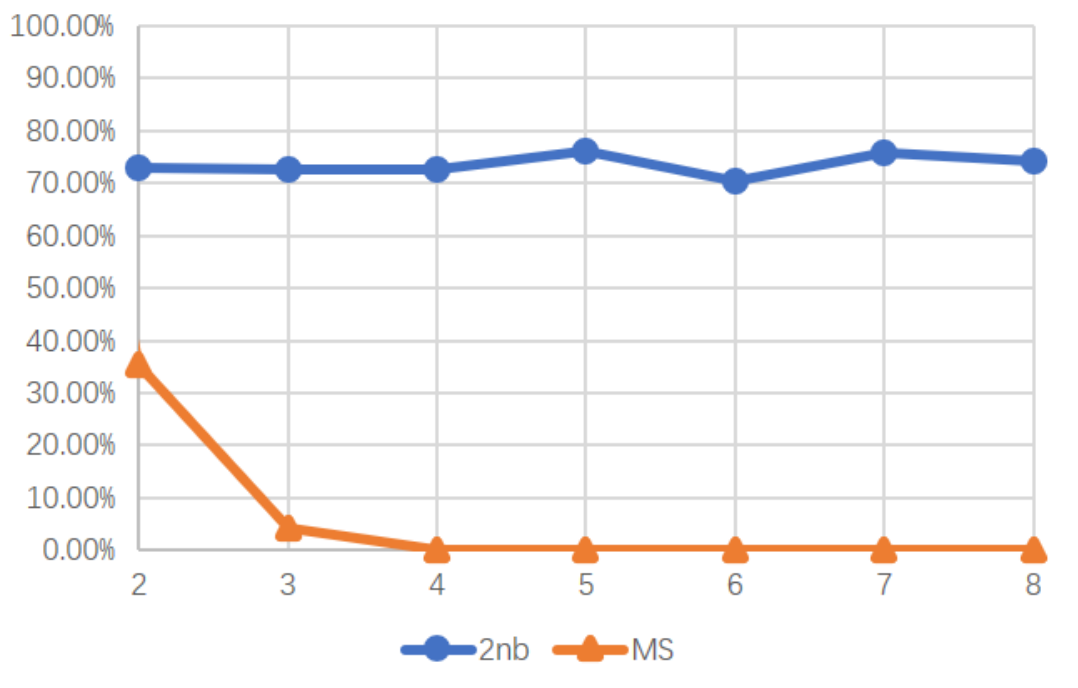}
			%		\caption{%${\mathcal{S}}_2$:
			%		Fairness: 8 dequeuers with slowdown factor $ 2^{i-1} $}
			\caption{\small{8 enqueuers + 8 dequeuers \\\hspace*{8mm}$y$-axis: \% of fair share obtained by dequeuer $i$ \\
					\hspace*{8mm}$x$-axis: dequeuer $i$, which is slowed by factor of $2^{i-1}$}}		
			\label{1248dq}
		\end{minipage}%
	\end{figure}
	
	In both settings, we observed that the fairness of the \DNB{2} algorithm is even better than in the system with two enqueuers and two dequeuers:
	in ${\mathcal{S}}_1$, even the enqueuer and dequeuer that were 8 times slower than the fastest enqueuer and dequeuer
	managed to complete 67\% and 77\% of their fair shares of operations, respectively (Figures~\ref{1234nq} and~\ref{1234dq});
	in ${\mathcal{S}}_2$, even the enqueuer and dequeuer that were 128 times slower (!)~than the fastest enqueuer and dequeuer
	completed 65\% and 76\% of their fair shares of operations, respectively (Figures~\ref{1248nq} and~\ref{1248dq}).

	In sharp contrast, we saw that the fairness of MS algorithm declines rapidly when the differences in speeds increase. %(Figures~\ref{1234nq},~\ref{1234dq},~\ref{1248nq},~\ref{1248dq}).
	As shown in Figures~\ref{1234nq} and~\ref{1234dq},
	enqueuers with a slowdown factor of 2, 3, 5 and 8 completed only 74\%, 33\%, 11\% and 3\% of their fair shares of enqueue operations, respectively;
	and
	dequeuers with a slowdown factor of 2, 3, 5 and 8 completed only 45\%, 13\%, 3\% and 0\%\footnote{Dequeuer 8 managed to complete only 5 dequeue operations, while the group of dequeuers completed about 50000 dequeues.}
	of their fair shares of dequeue operations, respectively.
	As shown in Figures~\ref{1248nq} and~\ref{1248dq},
	%In ${\mathcal{S}}_2$,
	enqueuers (dequeuers) with a slowdown factor of
	2, 4, 8 and 16 completed only 59\%, 15\%, 3\% and 0\% (35\%, 4\%, 0\% and 0\%)
	of their fair shares of enqueue (dequeue) operations, respectively.

	We also compared the throughput of the two algorithms, in terms of the total number of completed operations.
	As we can see in Table~\ref{throughputs}, the throughput of the MS algorithm was marginally higher,
	but by less than in the experiment with two enqueuers and two dequeuers:
	In setting ${\mathcal{S}}_0$,
	the 16 processes completed a total of
	223089 and 206605 operations with the MS algorithm and 
	the $\DNB{2}$ algorithm, respectively; so the throughput of the $\DNB{2}$ algorithm was 92.6\% of the MS algorithm.
	In  setting ${\mathcal{S}}_1$,
	they completed a total of
	112147 and 100490 operations with the MS algorithm and 
	the $\DNB{2}$ algorithm, respectively.
	%	so
	%	the throughput of the $\DNB{2}$ algorithm was 89.6\% of the MS algorithm.
	In  setting ${\mathcal{S}}_2$,
	they completed a total of
	111710 and 93357 operations with the MS algorithm and 
	the $\DNB{2}$ algorithm, respectively.
	%	so
	%	the throughput of the $\DNB{2}$ algorithm was 83.6\% of the MS algorithm.
	So in  ${\mathcal{S}}_1$ and ${\mathcal{S}}_2$,
	the throughput of the $\DNB{2}$ algorithm was 89.6\% and 83.6\% of the MS algorithm, respectively.

	\begin{table}[h]
		\vspace{5mm}
		\begin{center}
			\begin{tabular}{ |r|r|c|c|r|r| }
				\cline{3-6}
				\multicolumn{2}{c|}{} & NQ throughput & DQ throughput & Total & \DNB{2}/MS \\
				\hline
				\multirow{2}{*}{$ {\mathcal{S}}_0 $} & MS & 128167 & 94922 & 223089 & \multirow{2}{*}{ 92.61\% }\\\cline{2-5}
				& \DNB{2} & 93177 & 113428 & 206605 & \\
				\hline
				\multirow{2}{*}{$ {\mathcal{S}}_1 $} & MS & 61522 & 50625 & 112147 & \multirow{2}{*}{ 89.61\% }\\\cline{2-5}
				& \DNB{2} & 44445 & 56045 & 100490 & \\
				\hline
				\multirow{2}{*}{$ {\mathcal{S}}_2 $} & MS & 61003 & 50707 & 111710 & \multirow{2}{*}{ 83.57\% }\\\cline{2-5}
				& \DNB{2} & 42856 & 50501 & 93357 & \\
				\hline
			\end{tabular}
		\end{center}
		\caption{Throughputs of the MS and \DNB{2} algorithms in settings $ {\mathcal{S}}_0,~{\mathcal{S}}_1,~{\mathcal{S}}_2 $}
		\label{throughputs}
	\end{table}

	\vspace{2cm}

	\bibliography{biblio}

	\appendix
	
	\section{Proof of linearizability of the $\DNB{2}$ queue algorithm}
	\label{asec:appendix-algo-proof}
	
	Recall that $H$ is any history of the algorithm,
	and we constructed a completion $H'$ of $H$ as follows:
	\begin{itemize}
		\item For each incomplete $\NQ$ operation in $H$, 
		if it allocates a new node on line~\ref{A8} 
		and this node has been pointed to by the $\gTail$ pointer,
		then it is completed in $H'$ by returning $\vDone$;
		otherwise, it is removed from $H'$.
		
		\item For each incomplete $\DQ$ operation in $H$,
		if it allocates a new location on line~\ref{B10}
		and this location has been pointed to by $\gHead.\add$,
		then it is completed in $H'$ by returning the value of $\gHead.\val$
		at the earliest time when $\gHead.\add$ 
		points to the new location it allocated;
		otherwise, it is removed from $H'$. 
	\end{itemize}
	
	We then constructed a linearization $L$ of $H'$ as follows:
	\begin{itemize}
		\item Each $\NQ$ operation in $H'$ is linearized
		when $\gTail$ first points to the new node it allocates on line~\ref{A8}.
		
		\item Each $\DQ$ operation is linearized
		when $\gHead.\add$ first points to the new location 
		it allocates in line~\ref{B10}
		if the return value of the $\DQ$ is not $\emptyq$;
		otherwise it is linearized
		when the first process that sets $\gHead.\add$ 
		to point to the new location
		last reads $\gTail$ in line~\ref{D3}.
	\end{itemize}
	
	Note that by this construction, no operation can be linearized before it is invoked.
	This is easy to see for $\NQ$ operations 
	and $\DQ$ operations that do not return $\emptyq$,
	because these operations are linearized when
	their newly allocated node or location
	is pointed to by $\gTail$ or $\gHead.\add$.
	For each $\DQ$ operation that returns $\emptyq$,
	the process that first sets $\gHead.\add$ 
	to point to the new location on line~\ref{D15}
	must have previously read a pointer to that new location
	on line~\ref{B2} before executing line~\ref{D3},
	which is the linearization point of the operation.
	Clearly, a pointer to the new location cannot be created 
	before the location is allocated,
	so the operation is linearized after it is invoked.
	
	Thus we observe that:
	\begin{observation}\label{well-defined-L}
		The linearization points of the linearization $L$ of $H'$
		are well defined and within their execution intervals
		if the following properties hold:
		\begin{itemize}
			\item For each $\NQ$ operation,
			a pointer to the new node it allocates on line~\ref{A8}
			can be written into $\gTail$ at most once,
			and has occurred exactly once when the operation completes.
			
			\item For each $\DQ$ operation, 
			a pointer to the new location it allocates on line~\ref{B10}
			can be written into $\gHead$ at most once, 
			and has occurred exactly once when the operation completes. 
		\end{itemize}
	\end{observation} 
	
	\begin{lemma}\label{simple-prop}
		The following properties always hold:
		\begin{enumerate}
			\item The linked list starting from $\gINode$ 
			is never disconnected.\label{gINode-connected-list}
			
			\item $\gHead.\ptr$ always points to a node in the linked list
			starting from $\gINode$.\label{gINode-head-in-list}
			Furthermore, if $\gHead.\ptr$ is changed from pointing to a node $A$ 
			to pointing to a node $B$,
			then the $\textit{next}$ pointer of $A$ points to $B$.
			\label{gINode-delete-from-head}
			\label{gINode-head-move}
			
			\item $\gTail$ always points to a node in the linked list
			starting from $\gINode$ that is not before 
			the node pointed to by $\gHead.\ptr$.
			\label{gINode-tail-in-list-after-head}
			Furthermore, if $\gTail$ is changed from pointing to a node $A$ 
			to pointing to a node $B$,
			then the $\textit{next}$ pointer of $A$ points to $B$. 
			\label{gINode-tail-move} 
			
			\item Whenever a node $A$ is linked to a node $B$ 
			(in the sense that the $\textit{next}$ pointer 
			of $A$ is set to point to $B$),
			$A$ is the last node in the linked list
			starting from $\gINode$
			and $\gTail$ points to $A$.
			\label{gINode-add-to-tail} 
		\end{enumerate}
	\end{lemma}
	
	\begin{proof} 
		First observe that all these properties hold
		when the queue is initialized.
		From this, we show that each property must always hold as follows:
		\begin{enumerate}
			\item Since no step of the algorithm can change
			the $\textit{next}$ pointer of a node 
			from a non-$\mathnull$ value,
			links between nodes can never be removed.
			Thus the linked list starting from $\gINode$
			can never be disconnected, i.e.,
			Property~\ref{gINode-connected-list} always holds.
			
			\item $\gHead.\ptr$ can only be changed by 
			a successful CAS on line~\ref{D24}. 
			(Note that line~\ref{D15} cannot change $\gHead.\ptr$.)
			If the CAS in line~\ref{D24} is successful,
			then since line~\ref{D4} is executed before line~\ref{D24},
			observe that if $A$ is the node that $\gHead.\ptr$ originally pointed to
			and $B$ is the node that $\gHead.\ptr$ points to afterward,
			then the $\textit{next}$ pointer of $A$ points to $B$
			(because Property~\ref{gINode-connected-list} always holds).
			Thus $\gHead.\ptr$ is simply changed from one node to the next 
			in the linked list starting from $\gINode$. 
			So Property~\ref{gINode-head-move} always holds.
			
			\item The value of $\gTail$ can only be changed by 
			a successful CAS on one of lines \ref{C5.1}, \ref{C15} and \ref{C10}.
			If the CAS is successful,
			then since lines~\ref{C3repeat}, \ref{C3}, and \ref{C8} 
			are executed before lines 
			\ref{C5.1}, \ref{C15} and \ref{C10} respectively,
			observe that if $A$ is the node that $\gTail$ originally pointed to
			and $B$ is the node that $\gTail$ points to afterward,
			then the $\textit{next}$ pointer of $A$ points to $B$
			(because Property~\ref{gINode-connected-list} always holds).
			Thus $\gTail$ still points to a node in the linked list
			starting from $\gINode$.
			Moreover, $\gTail$ never lags behind $\gHead.\ptr$, 
			because line~\ref{D24},
			the only line that moves $\gHead.\ptr$ to the next node,
			is only executed when $\lHead.\ptr \ne \lTail$
			(line~\ref{D13}),
			where $\lHead = \gHead$ (line~\ref{D2})
			and $\lTail = \gTail$ (line~\ref{D3}).
			If $\gTail$ ever lags behind $\gHead.\ptr$, 
			then some operation must have moved $\gHead.\ptr$ 
			one step forward when $\lHead.\ptr = \lTail$, which never happens. 
			So Property~\ref{gINode-tail-in-list-after-head} always holds. 
			
			\item The only step of the algorithm that can link two nodes
			is a successful CAS on $\lTail\rightarrow\fNext$ on line~\ref{C8}.
			We need to show this $\lTail$ is pointing to
			the last node in the linked list starting from $\gINode$.
			Since Property~\ref{gINode-tail-in-list-after-head} always holds
			and $\lTail$ was set to $\gTail$ on line~\ref{C2},
			$\lTail$ points to a node in the linked list starting from $\gINode$.
			Since the CAS succeeds, the next pointer of this node is $\mathnull$.
			So since Property~\ref{gINode-connected-list} always holds,
			this node that that is pointed to by $\lTail$ 
			must be the last node of the linked list starting from $\gINode$.
			Furthermore, since Property~\ref{gINode-tail-move} always holds,
			this node is also still pointed to by $\gTail$.
			Thus Property~\ref{gINode-add-to-tail} always holds.  
			\qedhere
		\end{enumerate}
	\end{proof}

	\begin{definition}\label{Q-state}
		The state of a 2-nonblocking queue is defined to be 
		the linked list of nodes between the nodes pointed to 
		by $\gHead.\ptr$ and $\gTail$, 
		excluding the node pointed to by $\gHead.\ptr$.
	\end{definition}
	
	By Lemma~\ref{simple-prop}, the state of a queue is always well defined,
	because $\gHead.\ptr$ and $\gTail$ both point to nodes
	in the linked list of nodes starting from $\gINode$
	and $\gTail$ always points to a node that is not before the
	node pointed to by $\gHead.\ptr$.
	
	\begin{lemma}\label{node-linked-once-claim}
		The following two properties always hold: 
		\begin{enumerate}[\noindent(a)]
			\item $\gTail \rightarrow \fFlag $ is always $1$.
			\item Whenever a node is added to the list, it has $\fFlag=0$.
		\end{enumerate}
	\end{lemma}
	
	\begin{proof} 
		For (a) we note that $\gTail \rightarrow \fFlag$ is initially~1.
		Furthermore,
		$\gTail$ is changed only on lines~\ref{C5.1},~\ref{C15}, and~\ref{C10}.
		In each case, the preceding line has set to~$1$ the $\fFlag$ field
		of the node to which $\gTail$ will point.
		Finally, observe that after the $\fFlag$ field of this node is set to~$1$, 
		it can never be changed by any step of the algorithm.
		Thus (a) always holds.
		
		For (b) we note that the only place where a node is added to the list
		is on line~\ref{C8}.
		Suppose, for contradiction, that
		a process $p$ adds to the list a node $r$ with $\fFlag=1$
		in some execution of line~\ref{C8}, say at time $t_1$.
		The $\fFlag$ field of a node is set to~$1$
		only in lines~\ref{C5},~\ref{C14}, and~\ref{C9},
		and in all these cases the node has already been added to the list.
		Therefore, some process $p'$ added $r$ to the list before time $t_1$,
		say at time $t_2$.
		Let $t_0$ be the last time before $t_1$ when $p$ executed line~\ref{C2}.
		If $t_2$ is before $t_0$,
		then there are two cases: either $r\rightarrow\fFlag$ is set to $1$
		before $t_0$, or not.
		If $r\rightarrow\fFlag$ is set to $1$ before $t_0$,
		then the process $p$ that would add $r$ to the list at time $t_1$
		would, after time $t_0$, find that $r\rightarrow\fFlag = 1$ 
		on line~\ref{C4},
		and so not execute line~\ref{C8} at time $t_1$,
		contradicting that $p$ adds $r$ to the list at that time.
		Otherwise, since $r\rightarrow\fFlag$ was not set to $1$ before $t_0$
		and we have already proven that (a) always holds,
		$\gTail$ is also not moved to the added node $r$ before $t_0$ either.
		So the value of $\gTail \rightarrow \fNext$ that $p$ read after $t_0$
		is not $\mathnull$, but rather a pointer to the node $r$.
		But then, by line~\ref{C7}, $p$ does not execute line~\ref{C8} 
		at time $t_1$, contradicting that $p$ adds $r$ to the list at that time.
		
		If $t_2$ is between $t_0$ and $t_1$,
		then there are two cases. 
		In case 1, $p$ finds that the value of $\gTail \rightarrow \fNext$ 
		that $p$ read after $t_0$
		is not $\mathnull$, but rather a pointer to the node $r$.
		So again, by line~\ref{C7}, $p$ does not execute 
		line~\ref{C8} at time $t_1$,
		contradicting that $p$ adds $r$ to the list at that time.
		In case 2, $p$ finds that the value of $\gTail \rightarrow \fNext$ 
		that $p$ read after $t_0$
		is $\mathnull$, but since $t_2 < t_1$,
		$\gTail \rightarrow \fNext$ is changed from $\mathnull$
		to a pointer to $r$ at time $t_2$.
		So the CAS that $p$ executes at time $t_1$ fails,
		again contradicting that $p$ adds $r$ to the list at that time.
		%If, on the contrast, a process $ p $ enqueue a node with $ \fFlag=1 $ in one execution of line $ \ref{C8}. $ We call the time that process $ p $ executes line $\ref{C2}, \ref{C8}$ to be $ t_0 $ and $ t_1. $ The only possible condition that a process enqueues a node with $ \fFlag=1 $ is that the flag is set to $ 1 $ between $ t_0 $ and $ t_1. $ The node is linked to the queue at some time $ t_2 $ before $ t_1. $ (Since $ \fFlag=1 $ at $ t_1 $, some process sets the node's $ \fFlag$ to $ 1 $ before time $ t_1 $. $ \fFlag $ is set to $ 1 $ only on lines $ \ref{C5} $, $ \ref{C9} $, and $ \ref{C14} $, all of which set the flag of a node that is already linked to the queue. Thus the node is linked to the queue at some time $ t_2 $ before $ t_1 $.) If $ t_2 $ is between $ t_0 $ and $ t_1, $ either it finds $ \lNode\rightarrow\fNext \ne \mathnull $ on line $ \ref{C7}, $ or the process fails the CAS operation of line $ \ref{C8} $, so that node will not be enqueued. If $ t_2 $ is before $ t_0, $ when process $ p $ checks line $\ref{C7}$, $ \lNode\rightarrow \fNext $ should not be $ \mathnull $. Thus we have examined all the possible situations.
	\end{proof}	 
	
	\begin{lemma}\label{node-linked-once}
		Any node is linked in the queue at most once.
	\end{lemma}
	
	\begin{proof} 
		Suppose, for contradiction, that some node $r$ is linked 
		into the queue more than once.
		Consider the time just before $r$ is linked 
		into the queue for the second time.
		
		By Property~\ref{gINode-connected-list} of Lemma~\ref{simple-prop}
		and Definition~\ref{Q-state},
		$r$ is already within the 
		the linked list starting from $\gINode$ at this time.
		By Property~\ref{gINode-add-to-tail} of Lemma~\ref{simple-prop},
		$\gTail$ points to the last node in this linked list
		just before $r$ is linked into it again.
		Then, since Properties~\ref{gINode-connected-list} 
		and~\ref{gINode-tail-move} of 
		Lemma~\ref{simple-prop} always hold
		and $\gTail$ now points to the last node of the linked list,
		$\gTail$ has previously pointed to every node of the linked list.		
		Consequently, since $r$ is already within the linked list
		at this time, there must exist an earlier time 
		when $\gTail$ pointed to $r$.
		Thus by Lemma~\ref{node-linked-once-claim}(a),
		$r\rightarrow\fFlag = 1$ at this earlier time.
		
		Finally, observe that $r\rightarrow\fFlag$ can never be changed from $1$.
		Thus $r\rightarrow\fFlag$ still contains $1$ at the time
		when $r$ is linked into the queue for the second time
		--- contradicting Lemma~\ref{node-linked-once-claim}(b). 
	\end{proof}
	
	Now observe that by Lemmas~\ref{simple-prop} and\ref{node-linked-once}:
	
	\begin{corollary} \label{tail-node-most-once}
		The $\gTail$ can be set to point to each node at most once.
	\end{corollary}

	\begin{lemma}
		\label{ptr-in-gHead-once}
		Any new location allocated by a $\DQ$ operation 
		can not be pointed to by $\gHead.\add$ more than once.
	\end{lemma}
	
	\begin{proof}
		Suppose, for contradiction, that $\gHead.\add$ 
		points to a location $\add_0$ 
		allocated by a single $\DQ$ operation twice. 
		Let $t_1$ be the earliest time when
		a process $p$ executes line \ref{D15} or \ref{D24} and 
		succeeds in setting $\gHead.\add$ to point to $\add_0$. 
		Let $t_3$ be the earliest time after time $t_1$ when 
		a process $q$ (not necessarily different from $p$) 
		executes line~\ref{D15} or \ref{D24} and also
		succeeds in setting $\gHead.\add$ to point to $\add_0$.	
		Let $t_2$ be the last time $q$ reads $\gHead$ on line~\ref{D2}
		before $q$ sets $\gHead.\add$ to point to $\add_0$. 
		If $t_2$ is before $t_1$, then since process $p$ changes 
		the value of $\gHead$ at time $t_1$,
		at time $t_3$, $\lHead \ne \gHead$ for process $q$
		on line~\ref{D15} or~\ref{D24}, so $q$ fails its CAS,
		a contradiction.
		
		If $t_2$ is after $t_1$, process $q$ will read a $\gHead$ value 
		that has previously been set by $p$. 
		It is either a newer value or the value set by $p$. 
		In the latter case, process $q$ will execute line~\ref{D6} 
		to change the value in the location $\add_0$, 
		and then it finds the value in the location $\add_0$ is not $\mathnull$ 
		in line~\ref{D9}. 
		So $q$ will return on line~\ref{D10} 
		without executing line~\ref{D15} 
		or~\ref{D24}, a contradiction. 
		
		In the former case, there must be some other process(es) 
		performing successful CAS(es) on line \ref{D15} or \ref{D24} 
		to change $\gHead$ after $p$ changes $\gHead$. 
		At least one of them must see $\gHead.\add = \add_0$, 
		and executes line~\ref{D6} (and these happens between $t_1$ and $t_2$). 
		Then after time $t_2$, process $q$ will find that 
		the value in the location $\add_0$ is not $\mathnull$ 
		in line~\ref{D9}. 
		So $q$ will return on line~\ref{D10} 
		without executing line~\ref{D15} 
		or~\ref{D24}, a contradiction. 
	\end{proof}
	
	\begin{lemma}\label{flag-means-in-queue}
		If a node's $\fFlag$ field contains $1$,
		then the node is in the linked list of nodes starting from $\gINode$.
	\end{lemma}
	
	\begin{proof}
		The node $\gINode$ begins with $\fFlag = 1$,
		but all other nodes have $\fFlag$ set to $0$ when allocated 
		(lines~\ref{A8} to \ref{A11}).
		The $\fFlag$ field of a node can then be set to $1$
		only on lines~\ref{C5}, \ref{C14}, and \ref{C9}.
		
		If a node's $\fFlag$ field is set to $1$
		on line~\ref{C5} or~\ref{C14},
		then this node was previously
		the node after the node pointed to by $\gTail$
		(lines~\ref{C2} to \ref{C3} or lines~\ref{C2repeat} to \ref{C3repeat}).
		So by Lemma~\ref{simple-prop}, this node is
		in the linked list of nodes starting from $\gINode$.
		
		If a node's $\fFlag$ field is set to $1$
		on line~\ref{C9},
		then this node was previously linked on line~\ref{C8}
		to a node previously pointed to by $\gTail$ (line~\ref{C2}).
		So by Lemma~\ref{simple-prop}, this node is
		in the linked list of nodes starting from $\gINode$.
	\end{proof}
	
	\begin{lemma} \label{NQ-lin-once}
		For each $\NQ$ operation,
		$\gTail$ has been set to point to 
		the new node it allocates on line~\ref{A8}
		exactly once. 
	\end{lemma}
	
	\begin{proof}
		An $\NQ$ operation returns only after, 
		it has performed a $\trynq$ procedure on line~\ref{A12}
		with the node that it allocated on line~\ref{A8} 
		and the $\trynq$ procedure has returned $\vDone$ instead of $\vFail$.
		A $\trynq$ procedure returns $\vDone$ 
		either on line~\ref{C5.3} or line~\ref{C11}.
		
		In the former case,
		the $\trynq$ procedure previously found that 
		$\lNode \rightarrow \fFlag = 1$ (line~\ref{C4}),
		and then attempted to move $\gTail$ forward 
		in the linked list of nodes starting from $\gINode$ 
		if $\gTail$ was not already at the end of the list
		(lines~\ref{C2repeat} to \ref{C5.1}).
		By Lemma~\ref{flag-means-in-queue},
		$\lNode$ was already in the linked list of nodes
		starting from $\gINode$ when $\lNode \rightarrow \fFlag$ was set to $1$.
		By Property~\ref{gINode-add-to-tail} of Lemma~\ref{simple-prop},
		$\gTail$ was pointing to the node before $\lNode$
		when $\lNode$ was linked into the list.
		By Lemma~\ref{simple-prop},
		$\gTail$ can only be moved forward node by node through the list
		starting from $\gINode$ that can never be disconnected.
		Thus after the $\trynq$ procedure attempted to move $\gTail$ forward 
		in the linked list of nodes starting from $\gINode$ 
		if $\gTail$ was not already at the end of the list
		(lines~\ref{C2repeat} to \ref{C5.1}),
		observe that $\gTail$ must have previously 
		been set to point to $\lNode$.
		
		In the latter case,
		the $\trynq$ procedure previously 
		set the $\fNext$ field of some node 
		to point to $\lNode$ on line~\ref{C8},
		then attempted to change $\gTail$ 
		from $\lTail$ to $\lNode$ on line~\ref{C10}.
		By Property~\ref{gINode-add-to-tail} of Lemma~\ref{simple-prop},
		line~\ref{C8} added $\lNode$ to the end of 
		the linked list starting from $\gINode$,
		with $\gTail = \lTail$ pointing to the node
		that was previously at the end.
		
		Then observe that after line~\ref{C10},
		regardless of whether the CAS succeeds,
		$\gTail \ne \lTail$.
		So $\gTail$ has been changed after line~\ref{C8}.
		Thus by Property~\ref{gINode-tail-move} of Lemma~\ref{simple-prop},
		$\gTail$ must have been set to point to $\lNode$,
		the next node in the linked list, at some time 
		between lines~\ref{C8} and \ref{C11}.
		
		So in both cases, a $\trynq$ procedure
		only returns $\vDone$ after 
		$\gTail$ has previously 
		been set to point to $\lNode$.
		Thus by Corollary~\ref{tail-node-most-once},
		when an $\NQ$ operation completes,
		$\gTail$ has previously 
		been set to point to 
		the node it allocated on line~\ref{A8}
		exactly once.
	\end{proof}

	\begin{lemma} \label{DQ-lin-once}
		For each $\DQ$ operation,
		$\gHead.\add$ has been set to point to
		the new location it allocates on line~\ref{B10}
		exactly once. 
	\end{lemma}
	
	\begin{proof}
		An $\DQ$ operation returns only after, 
		it has performed a $\trydq$ procedure on line~\ref{B12}
		with the location that it allocated on line~\ref{B10} 
		and the $\trydq$ procedure did not return $\vFail$
		(line~\ref{B18}).
		A $\trydq$ procedure returns a non-$\vFail$ value 
		only on lines~\ref{D10},\ref{D16}, and \ref{D25}.
		
		If a $\trydq$ procedure returns on line~\ref{D10},
		then it has found that $\ladd \ne \mathnull$
		on line~\ref{D9}.
		A non-$\mathnull$ value can only be written to $\ladd$
		on line~\ref{D6}, and only if $\ladd$
		has been previously pointed to by $\gHead.\add$
		(line~\ref{D2}).
		Thus $\gHead.\add$ has previously pointed to $\ladd$
		when a $\trydq$ procedure returns on line~\ref{D10}.
		
		If a $\trydq$ procedure returns on line~\ref{D16} or line~\ref{D25},
		then it previously set $\gHead.\add$ to $\ladd$
		on line~\ref{D15} or \ref{D24}.
		
		So in all cases, a $\trydq$ procedure
		only returns a non-$\vFail$ value after 
		$\gHead.\add$ has previously 
		been set to point to $\ladd$.
		Thus by Lemma~\ref{ptr-in-gHead-once},
		when an $\DQ$ operation completes,
		$\gHead.\add$ has previously been set to point to 
		the location it allocated on line~\ref{B10} exactly once.  
	\end{proof}
	
	By Lemmas~\ref{tail-node-most-once}, \ref{ptr-in-gHead-once}, 
	\ref{NQ-lin-once} and \ref{DQ-lin-once},
	we have the following:
	\begin{itemize}
		\item For each $\NQ$ operation,
		a pointer to the new node it allocates on line~\ref{A8}
		can be written into $\gTail$ at most once,
		and has occurred exactly once when the operation completes.
		
		\item For each $\DQ$ operation, 
		a pointer to the new location it allocates on line~\ref{B10}
		can be written into $\gHead$ at most once, 
		and has occurred exactly once when the operation completes. 
	\end{itemize}
	Thus by Observation~\ref{well-defined-L}, 
	the linearization points of the linearization $L$ of $H'$ 
	are well defined and within their execution intervals.
	
	\begin{theorem}\label{seq-spec} 
		The linearization respects the sequential specification of a queue.
	\end{theorem}
	
	\begin{proof}
		By Definition~\ref{Q-state}, we regard the nodes 
		between $\gHead.\ptr$ and $\gTail$ as the actual nodes of the queue, 
		excluding the node pointed to by $\gHead.\ptr$. 
		By Lemma~\ref{simple-prop}, the queue is always connected 
		and there is no transient%\DMP{transition?} 
		state. 
		
		We linearize an $\NQ$ operation at the point 
		when $\gTail$ is moved forward through the linked list 
		to point to the node that the $\NQ$ operation allocated on line~\ref{A8}.
		Thus by Definition~\ref{Q-state},
		the allocated node is enqueued exactly at the linearization point
		of the $\NQ$ operation.
		
		We linearize a $\DQ$ operation that returns a non-$\emptyq$ value
		at the point when the location that it allocated on line~\ref{B10} 
		is pointed to by $\gHead.\add$ (line~\ref{D24}).
		Note that at this linearization point,
		$\gHead.\ptr$ would simultaneously be moved forward 
		through the linked list, and that the node which is 
		dequeued by Definition~\ref{Q-state}
		has value equal to the return value of the $\DQ$ operation. 
		
		We linearize a $\DQ$ operation that returns $\emptyq$
		at the point when the process that would on line~\ref{D15} set
		$\gHead.\add$ to point to the location it allocated on line~\ref{B10}
		last reads $\gTail$ on line~\ref{D3}.
		Note that at this linearization point,
		since $\lHead.\ptr = \lTail$ (line~\ref{D13}),
		the queue is empty by Definition~\ref{Q-state}.  
		
		Thus our linearization points are simultaneous 
		with the queue structure changes, 
		so they follow the semantics of a queue structure.
	\end{proof}
	\section{Universal 2-nonblocking algorithm}\label{appendix-universal}
	Below we give a \emph{universal construction} for 2-nonblocking objects of any type $T$.
	The input to this algorithm is the sequential specification of the type $T$
	given in the form of a function $\applyT$.
	This function maps tuples of the form $(\op,s)$
	where $\op$ is an operation to be applied to the object
	and $s$ is the object current state,
	to a tuple $(s',r)$, where $s'$ is the new state
	of the object and $r$ is the return value.
	As in our $\DNB{2}$ queue algorithm,
	a process first tries to help another process only \emph{once}
	(the ``altruistic'' phase in lines~\ref{2C2}-\ref{2C4}) and then it repeatedly tries
	to perform its \emph{own} operation
	until it is done (the ``selfish'' phase in lines~\ref{2C6}-\ref{2C13}).
	All the processes that need help compete by writing their call for help
	on the \emph{same} shared register $\gAnn$;
	as we explained earlier, this promotes fairness.
	The proof of correctness is a simplification of the one that we give
	for the $\DNB{2}$ queue algorithm.

	\makeatletter
	\renewcommand{\ALG@name}{Fig.}
	\makeatother 
	
	%\MP{Numbers must be MANUALLY set.}
	
	\setcounter{algorithm}{8}
	
	\begin{algorithm}[H]
		\caption{Shared objects for the universal 2-nonblocking algorithm}
		\label{alg8}
		%		\textbf{Shared Objects:} 
		\begin{algorithmic}[1]
			\Statex $ \gAnn $: A register with two fields:
			%			\Statex \hspace{4mm} $q$: a process id, initially arbitrary.
			\Statex \hspace{4mm} $\op$: the operation to execute.
			\Statex \hspace{4mm} $\add$: a pointer to a location that stores a returned value or $\mathnull$, initially pointing to any non-$ \mathnull $ value.
			\Statex $ \gR $: A compare\&swap with three fields:
			\Statex \hspace{4mm} $\stt$: a register containing an integer, initially with the initial state.
			\Statex \hspace{4mm} $\res$: the response value to be returned by the function, initially  $ \mathnull. $
			\Statex \hspace{4mm} $\add$: a pointer to a location where $ R.\res $ will be stored, initially pointing to any non-$ \mathnull $ value.
			\Statex
		\end{algorithmic}
		\vspace*{-4mm}
	\end{algorithm}
	
	\setcounter{algorithm}{1}
	\setcounter{figure}{1}
	
	\makeatletter
	\renewcommand{\ALG@name}{Algorithm}
	\makeatother
	
	%
	%
	%\begin{algorithm}[H]
	%	\caption{Shared objects for the universal 2-nonblocking algorithm}
	%	\label{alg7}
	%%		\textbf{Shared Objects:} 
	%		\begin{algorithmic}[1]
	%			\Statex $ \gAnn $: A register with three fields:
	%			\Statex \hspace{4mm} $q$: a process id, initially arbitrary.
	%			\Statex \hspace{4mm} $\op$: the operation to execute.
	%			\Statex \hspace{4mm} $\add$: a pointer to a location that stores a returned value or $\mathnull$, initially pointing to any non-$ \mathnull $ value.
	%			\Statex
	%			\Statex $ \gR $: A compare\&swap with three fields:
	%			\Statex \hspace{4mm} $\stt$: a register containing an integer, initially with the initial state.
	%			\Statex \hspace{4mm} $\res$: the response value to be returned by the function, initially  $ \mathnull. $
	%			\Statex \hspace{4mm} $\add$: a pointer to a location where $ R.\res $ will be stored, initially pointing to any non-$ \mathnull $ value.
	%			\Statex
	%		\end{algorithmic}
	%\end{algorithm}
	%	

	\begin{algorithm}[h]
		\caption{Universal 2-nonblocking algorithm (code for process $p$)}
		\begin{algorithmic}[1]
			\Procedure{\tnb}{$\op$}\label{2C1}
			%		\State $ (q,\op',\ladd):=\gAnn.\textsc{read}() $\label{2C2}
			\State $ (\op',\ladd):=\gAnn.\textsc{read}() $\label{2C2}
			
			%		\If{$ ^*(\ladd).\textsc{read}()==\mathnull $}\label{2C3}
			%		\State $\try(q,\op',\ladd) $\label{2C4}
			\State $\try(\op',\ladd) $\label{2C4}
			
			%		\EndIf \label{2C5}
			\State $ \ladd:= $ a pointer to a new location\label{2C6}
			\State $ ^*(\ladd).\textsc{write}(\mathnull) $\label{2C7}
			\Repeat\label{2C8}
			%		\State $ t:=\try(p,\op,\ladd) $\label{1C9}
			\State $ t:=\try(\op,\ladd) $\label{2C9}
			
			\If{$ t=\FALSE $}\label{2C10}
			%		\State $ \gAnn.\textsc{write}(p,\op,\ladd) $\label{2C11}
			\State $ \gAnn.\textsc{write}(\op,\ladd) $\label{2C11}
			\EndIf\label{2C12}
			\Until{$ t \ne \FALSE $}\label{2C13}
			\State \textbf{return} $ t $\label{2C14}
			\EndProcedure
			\Statex
			
			%		\Statex		
			%		\Procedure{\try}{$ q, \op, \ladd $}\label{2B1}
			\Procedure{\try}{$\op, \ladd $}\label{2B1}
			%\State $ v_1:=R_1.\textsc{read}(),v_2:=R_2.\textsc{read}(),\cdots,v_{s-1}:=R_{s-1}.\textsc{read}() $\label{2B2}
			\State $ \lR:=\gR.\textsc{read}() $\label{2B3}
			\State $ ^*(\lR.\add).\textsc{write}(\lR.\res) $\label{2B4}
			\If{$ ^*(\ladd).\textsc{read}() \ne \mathnull $}\label{2B5}
			\State \textbf{return} $ \ladd.\textsc{read}() $\label{2B6}
			\EndIf \label{2B7}
			%\State $ v',\res':= $ new value proposed based on $ v_1,v_2,\cdots,v_{s-1},\lR $ and $ q $\label{2B8}
			%		\State $( s',\res'):= \applyT(q,\op,\lR.\stt) $\label{2B8}
			\State $( s',\res'):= \applyT(\op,\lR.\stt) $\label{2B8}
			\If{$\CAS(\gR,\lR,(s',\res',\ladd))$}\label{2B9}
			\State \textbf{return} $ \res' $\label{2B10}
			\Else\label{2B11}
			\State \textbf{return} $ \FALSE $\label{2B12}
			\EndIf\label{2B13}
			\EndProcedure
		\end{algorithmic}
	\end{algorithm}
	
\end{document}